\documentclass[preprint,11pt]{elsarticle}
\journal{ }
\usepackage[usenames,dvipsnames]{xcolor}
\usepackage{color}
\usepackage[english]{babel}
\usepackage{graphicx}
\usepackage{subfigure}
\usepackage{pifont,latexsym,ifthen,amsthm,rotating,calc,textcase,booktabs}
\usepackage{extarrows}
\usepackage{amsfonts,amssymb,amsbsy,amsmath}
\usepackage[title]{appendix}
\numberwithin{equation}{section}
\newtheorem{theorem}{Theorem}[section]
\newtheorem{lemma}[theorem]{Lemma}
\newtheorem{corollary}[theorem]{Corollary}

\newtheorem{proposition}[theorem]{Proposition}

\newtheorem{problem}{RH problem}[section]

\newtheorem{main theorem}{Theorem}
\usepackage{natbib}
\usepackage{todonotes}
\usepackage{float}
\usepackage{subfigure}
\numberwithin{figure}{section}
\biboptions{sort&compress}
\usepackage{hyperref}
\usepackage{mathrsfs}
\hypersetup{
	colorlinks=true, 
	linktoc=all,     
	linkcolor=blue,  
	citecolor=red
}

\begin{document}
	
	
	
	\begin{frontmatter}
		\title{Soliton resolution and  asymptotic stability of $N$-loop-soliton  solutions  for the Ostrovsky-Vakhnenko equation}
		

		\author[inst2]{Ruihong Ma}

		\author[inst2]{Engui Fan$^{*,}$  }
		
		\address[inst2]{ School of Mathematical Sciences and Key Laboratory of Mathematics for Nonlinear Science, Fudan University, Shanghai,200433, China\\
			* Corresponding author and e-mail address: faneg@fudan.edu.cn  }
		
		
		
		

		\begin{abstract}

			The  Ostrovsky-Vakhnenko (OV)  equation
			\begin{align*}
				&u_{txx}-3\kappa u_x+3u_xu_{xx}+uu_{xxx}=0
			\end{align*}
			is a 	 short wave model of  the well-known  Degasperis-Procesi   equation
			and admits a  $3\times 3$ matrix  Lax pair.
			In this paper,  we study the soliton resolution and  asymptotic stability of $N$-loop soliton  solutions for the  OV equation with Schwartz initial data that supports soliton solutions.
			It is shown that the solution of the Cauchy problem can be characterized via a $3\times 3$ matrix  Riemann-Hilbert (RH)  problem in a new scale.
			Further by deforming the RH problem into   solvable models  with $\bar\partial$-steepest descent method, we obtain the soliton resolution  to the OV equation
			in two  space-time regions $x/t>0$ and $x/t<0$. This  result also  implies  that $N$-loop soliton solutions of the OV equation are asymptotically stable.
			\\[5pt]
		\end{abstract}
	\end{frontmatter}
	
	\tableofcontents
	\section{Introduction}
	In this paper, we study soliton resolution and asymptotic stability of $N$-loop soliton solutions  for the Ostrovsky-Vakhnenko   (OV) equation on the line
	\begin{align}
		&u_{txx}-3\kappa u_x+3u_xu_{xx}+uu_{xxx}=0,\label{eq:1.1}\\
		&u(x,0)=u_0(x),\quad (x,t)\in\mathbb{R}\times\mathbb{R}^+, \label{inieq:1.1}
	\end{align}
	where $\kappa>0$ is a parameter and  the initial value  $u_0$ is assumed in the Schwartz space $  \mathcal{S}(\mathbb{R})$ satisfying $1-u_{0xx}>0$.
	The OV equation \eqref{eq:1.1} arises in the theory of propagation of surface waves in deep water \cite{KL14}.
	For $\kappa=0$ and $\kappa=-1/3$, the OV equation \eqref{eq:1.1} respectively reduces to the derivative Burgers equation
	\begin{equation}
		(u_t+uu_{xx})_{xx}=0,
	\end{equation}
	and the differentiated Vakhnenko equation  \cite{EJ93}
	\begin{equation}\label{eq:1.4}
		(u_t+uu_{x})_x+u=0.
	\end{equation}

	The OV equation  \eqref{eq:1.1} is also  called  the short wave model of  the well-known  Degasperis-Procesi (DP) equation 
	\begin{equation}\label{eq:1.2}
		u_t-u_{xxt}+3\kappa u_x+4uu_x=3u_xu_{xx}+uu_{xxx},
	\end{equation}
	since under a scaling  transformation
	\begin{equation*}
		x\to \epsilon x,\quad t\to  t/\epsilon,\quad  u\to \epsilon^2 u,
	\end{equation*}
	then the DP equation \eqref{eq:1.2} becomes  OV equation  \eqref{eq:1.1} as $\epsilon\to0$.
	The DP equation was first discovered  by Degasperis and Procesi in   \cite{AM}.
	Afterward it was found that the DP equation  arises for modeling the propagation of shallow water waves over a flat bed in so-called
	moderate amplitude regime  \cite{RS,RI1, AD}.
	In recent years, there has been  attracted a lot of attention on  the DP equation \eqref{eq:1.2}
	due to its integrable structure and pretty mathematical
	properties \cite{J1,Y1}.   Liu and Yin proved the global existence and blow-up phenomena for the DP equation  \cite{YZ}.
	Constantin,  Ivanov and Lenells developed  the inverse scattering transform method \cite{ARIJ}  and  dressing transform method \cite{CI}  for the DP equation.
	Lenells proved  that the solution of  the initial-boundary value problem for the DP equation on the half-line
	can be expressed in term of  the solution of a RH problem  \cite{J2}.
	Later Boutet de Monvel,  Lenells and Shepelsky  successfully extended the Deift-Zhou  nonlinear steepest descent method  \cite{DP93}
	to the Cauchy problem \cite{AD2}  and  initial-boundary value  problem \cite{AJD}   to  obtain the long time asymptotics for  the   DP equation.
	Hou, Zhao, Fan and Qiao  constructed the algbro-geometric solutions  for the DP hierarchy
	\cite{HZF}.    Feola,   Giuliani and    Procesi developed the KAM theory close to an elliptic fixed point for quasi-linear Hamiltonian perturbations of the
	DP  equation on the circle    \cite{FGP}.

	In recent years, much work also has been done to study the various mathematical properties of the OV equation \eqref{eq:1.1}. For example, as an integrable system associated with
	$3\times 3$ matrix spectral problem,  the initial-boundary value problem for the OV equation on the half-line was investigated via the Fokas unified method \cite{JE16}. A bi-Hamiltonian formulation for the OV equation was established by using its higher-order symmetry and a new transformation to the Caudrey-Dodd-Gibbon-Sawada-Kotera equation \cite{JS13}. The shock solutions and singular soliton solutions, such as peakons, cuspons, and loop solitons, for the OV equation were constructed by developing the discontinuous Galerkin method \cite{QY19}. The well-posedness of the Cauchy problem for the OV equation and its relatives such as reduced Ostrovsky equation, generalized Ostrovsky equation  in Sobolev spaces has been widely studied using analytic techniques \cite{MD12,KK11,LA06,AY10}. It has been shown that the Cauchy problem for the OV equation has a unique global solution for initial data $u_0\in H^3(\mathbb{R})$  with $1-u_{0xx}(x) >0$  \cite{RD14}.
	
	Soliton resolution entails the remarkable property that  the solution  decomposes   into the sum of a finite $N$ number of separated solitons and a radiative part as $t\to\infty$.  The research of soliton resolution has been extensively studied in various wave equations.  For example,  soliton resolution along a sequence of times was investigated for the focusing energy critical wave equation \cite{TH17}. The  $\bar\partial$-steepest descent analysis  \cite{KTRPD1} is employed to analyze the soliton resolution for the NLS, derivative NLS  and mKdV equations  respectively \cite{MR18,RP18,GJ21}.
	In addition, the study  focused on the soliton resolution and large time behavior of solutions to the Cauchy problem for the Novikov equation with a nonzero background
  \cite{YE23}. These works contribute to the understanding of soliton resolution and its implications in different contexts.

	The Cauchy problem for integrable systems can be solved via the inverse scattering transform, which expresses the solution
	in terms of the solution of a matrix RH problem. However, the analysis process becomes difficult for higher-order
	matrix spectral problem.
	Recently Boutet de Monvel et al.  \cite{AD15} developed a RH  method to  the OV equation  \eqref{eq:1.1} with $3\times 3$ matrix spectral problem
	and  obtained its   smooth solitary wave solutions   under the reciprocal transformation
	\begin{equation}
		y(x,t):=x-\int_x^\infty((1-u_{ss}(s) )^{1/3}-1)ds.
	\end{equation}
	In the original variable $(x,t)$, the soliton solution of the OV equation \eqref{eq:1.1} is a multivalued function with a loop shape.
	Under the non-monotonic change of variable from $x\to y$, the soliton in the $(y,t)$ variable has a bell shape,
which is typical for solitons in integrable nonlinear evolution equations.
	In addition, for  the case without solitons,   the long-time  asymptotics for
	the OV equation  \eqref{eq:1.1} in  two asymptotic  regions ${\rm I}: x/t<0 $ and ${\rm II}: x/t>0$
	 were  derived  \cite{AD15}.
	
	In our work, we concern with the soliton resolution and  asymptotic stability of $N$-loop-soliton solutions for  the OV equation \eqref{eq:1.1}.
	Our results  have  the following two different  features.
	Firstly, we  investigate   long-time asymptotic behavior  for the OV equation \eqref{eq:1.1}  with  the initial data $u_0$
	which allows appearance  of  solitons.
	Secondly, we aim to   verify    the soliton resolution  for the OV equation \eqref{eq:1.1} via a $\bar\partial$-steepest descent method.
	This suggests that, at large  time,    solution of the Cauchy problem  (\ref{eq:1.1})-(\ref{inieq:1.1}) of the OV equation   will eventually break down into a combination of a radiation component and a finite number of  solitons.
	This  result   also implies  the asymptotic stability of $N$-soliton solutions of the OV equation \eqref{eq:1.1}, which
	so far has not been thoroughly explored.

	\subsection{Main results}
	
	The main result of this paper  is now presented as follows.
	\begin{theorem}\label{th1.1}
		Let
		$ \mathcal{  D} =\{r(z),\{\xi_n,c_n, \}_{n=1}^{6N}\}$  be the scattering data generated from associated  $u_0 \in \mathcal{S}(\mathbb{R})$.
		Denote by
		$u_{sol}(x,t;\hat{\mathcal{D}})$  as  the  soliton solutions   associated with   reflectionless scattering data
		$\hat{\mathcal{D}}=\{0,\{\xi_n,\hat c_n\}_{n=1}^{6N}\} $
		with
		\begin{equation}
			\hat c_n= c_n\exp\left( \frac{i}{\pi}\int_{{I}}\frac{\log(1-|r(s)|^2)}{s-\xi_n}ds\right).
		\end{equation}
		Then the solution $u(x,t)$ to the  Cauchy problem (\ref{eq:1.1})-(\ref{inieq:1.1}) has the following asymptotics behavior:
		\begin{itemize}
			\item[$\blacktriangleright$]{\bf In the region I: $x/t<0$}, we have
			\begin{align}
				&u(x,t)
				 =  {u}_{sol}(y,t;  {\mathcal{D}})+t^{-1/2}f_t(y,t)+\mathcal{O}(t^{-\frac{3}{4}}),\nonumber\\
				& x:= x(y,t)=y+g(y,t)+t^{-1/2}f(y,t)+\mathcal{O}(t^{-3/4}),\nonumber
			\end{align}
where   $g(y,t)$ is expressed as shown in \eqref{eqg} and
$f(y,t)$ is an interaction between dispersion and soliton given by \eqref{eqf}.			
			
			\item[$\blacktriangleright$]{\bf In the region II: $x/t>0$}, we have
			 \begin{align}
				&u(x,t)
				 =u_{sol}(y,t; {\mathcal{D}}  )+\mathcal{O}(t^{-1}),\nonumber\\
				&x:=x(y,t)=y+g(y,t)+\mathcal{O}(t^{-1}), \nonumber
			\end{align}
			where   $g(y,t)$ can be derived  by \eqref{eqg}.
		\end{itemize}
	\end{theorem}
	
	As a  corollary of Theorem \ref{th1.1}, we obtain   the asymptotic stability of an $N$-loop soliton at the same time.
	
	\begin{corollary}\label{th1.2}
		Given $N$-soliton  $u_{sol}(x,t;	 \hat{\mathcal{D}})$
		with  $ \xi^{sol}_n,n=1,\dots,6N$. There exist positive constants $\eta_0=\eta_0(u_{sol})$, $T=T(u_{sol})$, and $K=K(u_{sol})$ such that for any initial data $u_0\in\mathcal{S}(\mathbb{R})$ with
		$$\eta_1:=|u_0-u_{sol}(x,0; \hat{\mathcal{D}})|\leq\eta_0,$$
	 such that
		\begin{equation}
		\sum_{j=1}^{6N}(| \xi_n-\xi_n^{sol}|+|c_n-\hat c_n|)\leq K\eta_0,
		\end{equation}
		then the solution of the Cauchy problem \eqref{eq:1.1}-(\ref{inieq:1.1})  asymptotically separates into a sum of $N$ one-loop soliton solution
		\begin{equation}
			\underset{x\in\mathbb{R}}{\sup}\,\left|u(x,t)-\sum_{n=1}^{6N}\mathcal{Q}_{sol}(x,t,\xi_n)\right|\leq K\eta_0 t^{-1},\quad t>T,
		\end{equation}
		where single-soliton $\mathcal{Q}_{sol}(x,t,\xi_n)$ is given by  \eqref{eq:1.10}.
	\end{corollary}
	
	\subsection{Plan  of the proof}
	In Section \ref{sec2},  starting for the spectral problem (\ref{eq:2.1}), we  analyze  the analyticity, symmetries, and asymptotic behavior of the Jost functions and scattering data.  Subsequently,  we establish a  RH problem  \ref{RH2.1}   associated with the Cauchy problem (\ref{eq:1.1})-(\ref{inieq:1.1}), which is used 	   to capture  the leading-order asymptotics of the solution for the OV equation
	   with the $\bar\partial$-steepest descent  analysis.
	    In Section  \ref{sec33},  we    show  soliton resolution for the OV equation in
	     the asymptotic   region  $x/t<0$
	 in which   there are  six phase points in the jump contours.
	The main contribution to the RH problem
	 comes from jumps near phase points and solitons.
	In Section \ref{sec44} ,  we show  soliton resolution for the OV equation
	    in the  asymptotic region  $x/t>0$
	 in which      the jump contours admit no phase points.

	\section{Inverse  Scattering Transform and RH Problem} \label{sec2}

	In this section, we  state some  basic
	results on  the  inverse scattering transform and the RH problem  associated with the Cauchy problem
	(\ref{eq:1.1})-(\ref{inieq:1.1}).   The details   can be found in
	\cite{AD15}.

	\subsection{Spectral analysis on  Lax pair }\label{sec2.2}
	The OV equation \eqref{eq:1.1}  admits a scalar    Lax pair \cite{AD15,AJ02}
	\begin{subequations}\label{eq:2.1}
		\begin{align}
			&\psi_{xxx}=\lambda(-u_{xx}+\kappa)\psi,\\
			&\psi_t=\frac{1}{\lambda}\psi_{xx}-u\psi_x+u_x\psi.
		\end{align}
	\end{subequations}
	Define a vector-valued function
	$\Phi:=\Phi(x,t;z)=\begin{pmatrix}
		\psi,\psi_x,\psi_{xx}
	\end{pmatrix}^T$, then  (\ref{eq:2.1}) can be written in matrix form
	\begin{align}\label{eq:2.2}
		&\Phi_x=\begin{pmatrix}
			0&1&0\\
			0&0&1\\
			z^3q^3&0&0
		\end{pmatrix}\Phi,\quad\Phi_t=\begin{pmatrix}
			u_x&-u&z^{-1}\\
			1&0&-u\\
			-z^3uq^3&1&-u_x
		\end{pmatrix}\Phi,
	\end{align}
	where $z$ is defined by
	$\lambda=z^3$ and $q=(1-u_{xx})^{1/3}>0$.	 Without loss of generality, we take $\kappa=1$ in the Lax pair (\ref{eq:2.1})  in   our paper.
	
	In order to have better control of the behavior of solutions of system \eqref{eq:2.2} for large $z$, define
		\begin{align}\label{eq:2.3}
			&D(x,t)=\begin{pmatrix}
				q^{-1} &0&0\\
				0&1&1\\
				0&0&q
			\end{pmatrix},\,\,P(z)=\begin{pmatrix}
				1&1&1\\
				\lambda_1&\lambda_2&\lambda_3\\
				\lambda_1^2&\lambda_2^2&\lambda_3^2\\
			\end{pmatrix},
		\end{align}
	where $\omega=e^{2\pi i/3}$, $\lambda_j =z\omega^j, j=1,2,3$.
	Then under the   transformation
	\begin{equation}\label{eq:2.4}
		\tilde{\Phi}=P^{-1}D^{-1}\Phi,
	\end{equation}
 the system \eqref{eq:2.2} is changed  into another Lax pair
	\begin{subequations}\label{eq:2.5}
		\begin{align}
			&\tilde{\Phi}_x-q\Lambda\tilde{\Phi}=U\tilde{\Phi},\\
			&\tilde{\Phi}_t+(uq\Lambda -\Lambda^{-1} )\tilde{\Phi}=V\tilde{\Phi},
		\end{align}
	\end{subequations}														
	where 																																															
	\begin{subequations}
		\begin{align*}
			&\Lambda=\begin{pmatrix}
				\lambda_1 &0&0\\
				0&\lambda_2&0\\
				0&0&\lambda_3
			\end{pmatrix},\,\,\,\,\,U=\frac{q_x}{3q}\begin{pmatrix}
				0&1-\omega^2&1-\omega\\
				1-\omega&0&1-\omega^2\\
				1-\omega^2&1-\omega&0
			\end{pmatrix},\\
			&V=-uU+\frac{1}{3z}\left\{3\left(\frac{1}{q}-1\right)I+\left(q^2-\frac{1}{q}\right)\begin{pmatrix}
				1&1&1\\
				1&1&1\\
				1&1&1
			\end{pmatrix}\right\}\begin{pmatrix}
				\omega^2&0&0\\
				0&\omega&0\\
				0&0&1
			\end{pmatrix}.
		\end{align*}
	\end{subequations}
	
	The transformation (\ref{eq:2.4}) makes $U$ and $V$  be  bounded at $z=\infty$, which is appropriate for controlling the behavior of its solutions for large $z$. Observing the system (\ref{eq:2.5}),   we define a  reciprocal   transform
	\begin{align}
		&y(x,t):=x-\int_x^\infty(q(s,t)-1)ds,\label{eq:2.7}
	\end{align}
	and   a  3$\times$3 matrix  function
	\begin{equation}\label{eq:2.9}
		\Psi=\tilde{\Phi}e^{-Q},
	\end{equation}
	where  $Q=y(x,t)\Lambda +t\Lambda^{-1} $. Then (\ref{eq:2.5}) reduces to the system
	\begin{subequations}\label{eq:2.10}
		\begin{align}
			&\Psi_x -[Q_x,\Psi ]=U\Psi,\\
			&\Psi_t -[Q_t,\Psi ]=V\Psi,
		\end{align}
	\end{subequations}
	whose  solutions   can be formed by
	\begin{equation}\label{eq:2.15}
		\Psi_{jl}=I_{jl}+\int_{\infty_{jl}}^x e^{-\lambda_j(z)\int_x^s q(\eta,t)d\eta}[U\Psi (s,t;z)]_{jl}e^{\lambda_l(z)\int_x^s q(\eta,t)d\eta}ds,
	\end{equation}
	where  $\infty_{jl}$ defined as
	\begin{equation}\label{eq:2.11}
		\infty_{jl}=\begin{cases}
			+\infty,\quad {\rm Re}\,\lambda_j(z)\ge{\rm{ Re}} \,\lambda_l(z),\\
			-\infty,\quad {\rm Re}\,\lambda_j(z)<{\rm{ Re}}\,\lambda_l(z).
		\end{cases}
	\end{equation}

	Define  the rays
	\begin{equation}\label{eq:2.12}
		 l_n= e^{\frac{\pi i}{3}(n-1)}\mathbb{R}^+, \  \ n=1,\cdots,6,
	\end{equation}
	which  divide  the $z$-plane into six sectors
	\begin{equation}
		\Omega_{n}=\left\{z\,|\,\frac{\pi}{3}(n-1)<\arg\, z<\frac{\pi}{3}n\right\},\quad n=1,\dots,6.
	\end{equation}
	
	It has been demonstrated that the  Jost function $\Psi(z)$ defined by equation  \eqref{eq:2.9}  possesses the subsequent properties  \cite{AD15}.
	
	\begin{proposition}\label{pro2.2}
		The system of  integral equations \eqref{eq:2.15}  determines 3$\times$3-matrix valued solution $\Psi(z)$ of system \eqref{eq:2.10} having the following properties
		\begin{itemize}
			\item $\det \Psi(z)\equiv1$ and $\Psi(z)\to I$ as $z\to\infty$.
			\item The function $\Psi(z)$ is analytic in $\mathbb{C}\setminus\Sigma$, where $\Sigma=\bigcup_{n=1}^6 l_n$.
			\item  $\Psi(z)$ satisfies the symmetry relations:
			$$\Psi(z)=\Gamma_1\overline{\Psi(\bar{z})}\Gamma_1=\Gamma_2\overline{\Psi(\bar{z}\omega^2)}\Gamma_2=\Gamma_3\overline{\Psi(\bar{z}\omega)}\Gamma_3=\Gamma_4\Psi(z\omega)\Gamma_4^{-1},$$
			with
			$$\Gamma_1=\begin{pmatrix}
				0&1&0\\
				1&0&0\\
				0&0&1
			\end{pmatrix},\Gamma_2=\begin{pmatrix}
				0&0&1\\
				0&1&0\\
				1&0&0
			\end{pmatrix},\Gamma_3=\begin{pmatrix}
				1&0&0\\
				0&0&1\\
				0&1&0
			\end{pmatrix},\Gamma_4=\begin{pmatrix}
				0&0&1\\
				1&0&0\\
				0&1&0
			\end{pmatrix}.$$
		\end{itemize}
	\end{proposition}
	
	 $\Psi(z)$ is continuous boundary value $\Psi_\pm(z)$ on
	$ \Sigma$,  which  are  related by
	\begin{equation}\label{eq:2.27}
		\Psi_+(z)=\Psi_-(z) e^{Q }S_0(z) e^{-Q },\quad z\in\Sigma,
	\end{equation}
	where  the scattering matrix $S_0(z)$  is   given by
	\begin{align}
&S_0(z)=\begin{pmatrix}
				1&\overline{r(z)} &0\\
				-r(z)&1-|r(z)|^2&0\\
				0&0&1
			\end{pmatrix}, \,\,\,\,\,\,\,\,\,\ z\in l_1\cup l_4, \nonumber \\
			&S_0(z)=\begin{pmatrix}
				1&0&0\\
				0&1 & \overline{r(\omega z)} \\
				0&-r(\omega z)&1-|r(\omega z)|^2
			\end{pmatrix},  \,\,\,\ z\in l_2\cup l_5, \nonumber\\
	&S_0(z)=\begin{pmatrix}
				1-|r(\omega^2 z)|^2 &0&-r(\omega^2z)\\
				0&1 & 0 \\
				\overline{r(\omega^2 z)}&0 &1
			\end{pmatrix}, \ z\in l_3\cup l_6,	\nonumber
	\end{align}
	and reflection coefficient $r(z)$ is determined by the initial data  $u_0$.

\subsection{Distribution of poles}

	According to \cite{AD15}, $\Psi(z)$ has at most finite number of simple poles,
and the poles for  distinct columns  $\Psi_1(z),\Psi_2(z)$ and $\Psi_3(z)$ may lie on the rays $i\mathbb{R},  e^{i\pi/6}\mathbb{R}$ and  $ e^{-i\pi/6}\mathbb{R}$, respectively.
There are two types of  poles, denoted by the sets $\{z_n\}_{n=1}^{N_1}$ and $\{\ell_n\}_{n=1}^{N_2}$, which may be placed along $e^{i\pi/6}\mathbb{R}^+$, see Figure \,\ref{F2111}. It is convenient to denote
	$$\xi_n:=z_n,\,n=1,\dots,{N}_1, \ \ \xi_{{N}_1+n}:=\ell_n,\,n=1,\dots,{N}_2, $$
and $N=N_1+N_2.$	By symmetry property from Proposition \ref{pro2.2},   the other poles of $\Psi(z)$ are
\begin{align}
&\xi_{n+N}=\omega\bar\xi_n,\,\,\,\,\,\,\,\,\,\,\xi_{n+2N}=\omega\xi_n,\,\,\,\,\,\, \xi_{n+3N}=\bar\xi_{n+2N}, \nonumber \\
&\xi_{n+4N}=\bar\xi_{n+N}, \ \ \xi_{n+5N}=\bar\xi_n, \,\,\,\,\quad n=1,\cdots,N.\nonumber
\end{align}	
	Therefore, the discrete spectrum is
$$
	\mathcal{Z}=\mathcal{Z}^+\cup\bar{\mathcal{Z}}^+,\,\,\,\, \mathcal{Z}^+=\{\xi_n\,\}_{n=1}^{3N}\subset\mathbb{C}^+.
$$
	We denote norming constants $c_n\in\mathbb{C}^*$ for $\xi_n\in e^{i\pi/6}\mathbb{R}^+$, then
	other norming constants are given by
\begin{align}\nonumber
&c_{n+N}=\bar c_n\omega, \,\,\,\,\,\,\,\,\,\, c_{n+2N}= c_n\omega,   \,\,\,\,\, c_{n+3N}=\bar c_{n+2N}, \\
& c_{n+4N}=\bar c_{n+N},\,\,\,\,
 c_{n+5N}=\bar c_n , \quad\,\ n=1,\cdots,N.
\end{align}		
For initial data $u_0\in \mathcal{S}(\mathbb{R})$, the collection
	\begin{equation}\label{eq:2.34}
		\mathcal{D}=\{r(z),\{\xi_n,c_n\}_{n=1}^{6N}\}
	\end{equation}
	is called the scattering data.
	The time evolution of the scattering data $\mathcal{D}$ is given by
	\begin{equation*}
		\mathcal{D}(t)=\{r(z)e^{2it \theta(\xi_n)},\{\xi_n,c_n(t)e^{2it \theta(\xi_n)}\}_{n=1}^{6N}\}.\end{equation*}
where  $\theta(\xi_n)=-\frac{\sqrt{3}}{2}\Big(\frac{y}{t} \xi_n-\frac{1}{\xi_n}\Big)$.
	The  map $\mathcal{D}(t)\mapsto u(x,t)$ seeks to recover the solution of equation \eqref{eq:1.1}-\eqref{inieq:1.1} from its scattering data.

		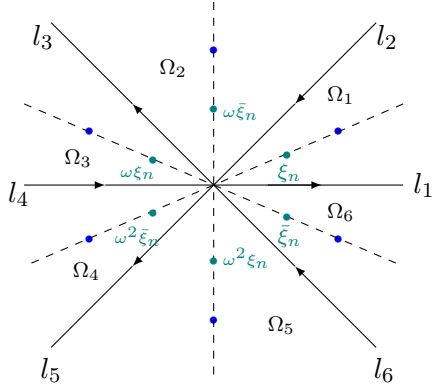
\begin{figure}[htp]
		\begin{center}
			\begin{tikzpicture}[scale=0.36]
				[node distance=1cm]
				\draw [-latex](-7,0)--(-4,0);
				\draw [-latex](-4,0)--(4,0);
				\draw[-](2,0)--(7,0)node[right]{$l_1$} ;	
				\draw [-latex](6,-6)--(3,-3);
				\draw [-latex](3,-3)--(-3,3);
				\draw[-](-3,3)--(-6,6);	
				\draw [-latex](6,6)--(3,3);
				\draw [-latex](3,3)--(-3,-3);
				\draw[-](-3,-3)--(-6,-6);	
				\draw(-6,-6)node[below] {$l_5$};
				\draw(-6.3,6.3)node[below] {$l_3$};
				\draw(6.4,6.3)node[below] {$l_2$};
				\draw(-7.2,0.5)node[below] {$l_4$};
				\draw(6.3,-6)node[below] {$l_6$};
				\node [blue]at(-4.6,2) {$\scriptscriptstyle\bullet$};
					\node [blue]at(4.6,2) {$\scriptscriptstyle\bullet$};
				\node [blue]at(-4.6,-2) {$\scriptscriptstyle\bullet$};
				\node [blue]at(4.6,-2) {$\scriptscriptstyle\bullet$};
				\node [blue]at(0,5) {$\scriptscriptstyle\bullet$};
				\node [blue]at(0,-5) {$\scriptscriptstyle\bullet$};
				\draw(4.7,4)node[below] {$\scriptstyle\Omega_1$};
			\draw(-5,1.7)node[below] {$\scriptstyle\Omega_3$};
			\draw(-4.7,-2.5)node[below] {$\scriptstyle\Omega_4$};
			\draw(4.7,-0.3)node[below] {$\scriptstyle\Omega_6$};
			\draw(-1.5,5)node[below] {$\scriptstyle\Omega_2$};
			\draw(2.5,-4.5)node[below] {$\scriptstyle\Omega_5$};
				\draw [dashed](7,3)--(-7,-3);
				\draw [dashed](-7,3)--(7,-3);
				\draw[dashed](0,-7)--(0,7);
				\node [teal]at(2.7,1.1) {$\scriptscriptstyle\bullet$};
				\node [teal]at(-2.25,0.9) {$\scriptscriptstyle\bullet$};
				\node[teal]at(-2.25,-1.05){$\scriptscriptstyle\bullet$};
				\node[teal]at(2.7,-1.2){$\scriptscriptstyle\bullet$};
				\node[teal]at(0,2.8){$\scriptscriptstyle\bullet$};
				\node[teal]at(0,-2.8){$\scriptscriptstyle\bullet$};
				\node [teal]at (2.8,0.45) {$\scriptstyle \xi_n$};
				\node [teal]at (-2.8,0.45) {$\scriptscriptstyle \omega\xi_n$};

				\node [teal]at (2.8,-1.9) {$\scriptstyle \bar{\xi}_n$};
				
				\node [teal]at (-2.8,-1.9) {$\scriptscriptstyle\omega^2\bar{\xi}_n$};
				
				\node [teal]at (1.2,-2.7) {$\scriptscriptstyle \omega^2\xi_n$};
				
				\node [teal]at (1,2.7) {$\scriptscriptstyle \omega \bar\xi_n$};
			\end{tikzpicture}
		\end{center}
		\caption{\footnotesize  The $z$-plane is divided into six analytical  domains $\Omega_n, n=1,...,6$ by
			six   rays  $l_n$ and the distribution of two kinds of poles are denoted by  $z_n$ $(\textcolor{blue}{\bullet})$ and $\ell_n$ ($\textcolor{teal}{\bullet}$).}
		\label{F2111}
	\end{figure}

	\subsection{Set up of a RH problem}\label{sec2.6}
	
By using  the variable $y=y(x,t)$ by \eqref{eq:2.7}, we define
	\begin{equation}
		M(y,t;z):=\Psi(x(y,t),t;z),
	\end{equation}
	which satisfies  following RH problem.
	\begin{problem}\label{RH2.1}
		Find a matrix  function $M(z):\mathbb{C}\setminus(\Sigma\cup\mathcal{Z})\to SL_3(\mathbb{C})$ with the following properties
		
		$\blacktriangleright$ \emph{ { The normalization condition:}}
		$$M(z)=I+\mathcal{O}(z^{-1})\quad z\to\infty.$$
		
		$\blacktriangleright$\emph{ {The symmetry proposition:}}	
		$$M(z)=\Gamma_1\overline{M(\bar z)}\Gamma_1=\Gamma_2\overline{M(\omega^2\bar z)}\Gamma_2=\Gamma_3\overline{M(\omega \bar z)}\Gamma_3.$$

		$\blacktriangleright$ 	 ${M}(z)$  satisfies  the  jump relation:\begin{equation}\label{eq:235}
			 M_+( z) = M_-( z) J( z),\quad z\in \Sigma,
		\end{equation}
		where
		\begin{align*}
			J( z)&=\begin{pmatrix}
				1&\bar{r}(z)e^{-2it\theta(z)}&0\\
				-r(z)e^{2it\theta(z)}&1-|r(z)|^2&0\\
				0&0&1
			\end{pmatrix},\qquad\qquad\,\,\,\,\,  z\in\mathbb{R},\\
			&=\begin{pmatrix}
				1-|r(\omega^2z)|^2&0&-r(\omega^2z)e^{2it\theta(\omega^2z)}\\
				0&1&0\\
				\bar{r}(\omega^2z)e^{-2it\theta(\omega^2z)}&0&1
			\end{pmatrix},\quad z\in\omega\mathbb{R},\\
			&=\begin{pmatrix}
				1&0&0\\
				0&1&	\bar{r}(\omega z)e^{-2it\theta(\omega z)}\\
				0&	-r(\omega z)e^{2it\theta(\omega z)}&1-|r(\omega z)|^2
			\end{pmatrix},\qquad\,\,\,\, z\in\omega^2\mathbb{R},
		\end{align*}
		where  $\theta(z)=-\frac{\sqrt{3}}{2}\Big(\xi z-\frac{1}{z}\Big)$ with $\xi=y/t$  is  the phase function.
		
		$\blacktriangleright$  $M( z)$  has simple poles   $\xi_n\in\mathcal{Z}^+, \ \bar\xi_n\in \bar{\mathcal{Z}}^+, \ n=1,\dots 3N$ at which
		\begin{align}\nonumber
		&	\underset{z=\xi_n\phantom{+44}}{\rm\,Res}\,M( z) =\underset{z\to \xi_n}{\lim}M(z){V}(z), \\\nonumber
			&\underset{z= \xi_{n+N}\phantom{2}
   }{\rm\,Res}\,M( z) =\underset{z\to \xi_{n+N}}{\lim}M(z)\Gamma_3\bar{V}(z)\Gamma_3, \\	\nonumber
				&\underset{z=\xi_{n+2N}}{\rm\,Res}\,M( z) =\underset{z\to \xi_{n+2N}}{\lim}M(z)\Gamma_1{V}(z)\Gamma_1, \\\nonumber
			&\underset{z=\bar\xi_n\phantom{+2n}}{\rm\,Res}\,M( z) =\underset{z\to\bar\xi_n\phantom{+2}}{\lim}M( z)\Gamma_4^{-1}\bar{{V}}(z) \Gamma_4,\\\nonumber
				&\underset{z=\bar\xi_{n+N}\phantom{n}
    }{\rm\,Res}\,M( z) =\underset{z\to\bar\xi_{n+N}}{\lim}M(z)\Gamma_4^{-1}\Gamma_3\bar{V}(z)\Gamma_3\Gamma_4,\\\nonumber
				&\underset{z=\bar\xi_{n+2N}}{\rm\,Res}\,M( z) =\underset{z\to \bar\xi_{n+2N}}{\lim}M(z)\Gamma_4^{-1}\Gamma_1\bar {V}(z)\Gamma_1\Gamma_4,
		\end{align}
		where ${V}(\xi_n)$ is the nilpotent matrix defined by
			\begin{align*}
	&{V}(\xi_n)=\begin{pmatrix}
			0&0&0\\
			c_n e^{2it\theta(\xi_n)}&0&0\\
			0&0&0
		\end{pmatrix},	\  \  n=1,\dots, N_1, \\
	&{V}(\xi_n)=\begin{pmatrix}
			0&0&0\\
			0&0&c_n e^{-2it\theta(\omega \xi_n)}\\
			0&0&0
		\end{pmatrix}, \  \  n= N_1+1,\dots,N.
	\end{align*}
	\end{problem}
	
The   solution $u(x,t)$ of the OV equation \eqref{eq:1.1} can be described as the limit $z\to0$ with respect to the solution $M(z)$ of the RH problem \ref{RH2.1} in the following way
	\begin{align}\label{eq:2.39}
		u(x,t)= x_t(y,t),
	\end{align}
	where
\begin{align}
	&x:= x(y,t)=y+\frac{\sum_{j=1}^3{ M}'_{j3}(0)}{\sum_{j=1}^3{M}_{ j3}(0)}.
	\label{eq:2.40}
\end{align}



\subsection{Classification of asymptotic regions}\label{sec2.8}

 According to the numbers of stationary phase   points on the contour $\Sigma$,
 the classification of asymptotic regions is given as follows.

	\begin{enumerate}
		\item[$\blacktriangleright$] For the range $\xi<0$, there  exist six  phase   points on $\Sigma$:
		$$\varkappa_{nj} := (-1)^j \omega^n\varkappa, \ j=0,1;n=0,1,2,$$
		where  $\varkappa=1/\sqrt{|\xi|}$, see Figure \ref{sig2}. The soliton resolution for the OV equation  in this case will be presented in Section \ref{sec33}.
\begin{figure}
	\subfigure[  $e^{2it\theta(z)}\to0$ as $t\to \infty$ in  the purple region, $e^{-2it\theta(z)}\to0$ as $t\to \infty$ in  the white region]{\label{a}
		\begin{minipage}[t]{0.25\linewidth}
			\centering
			\includegraphics[width=1.5in]{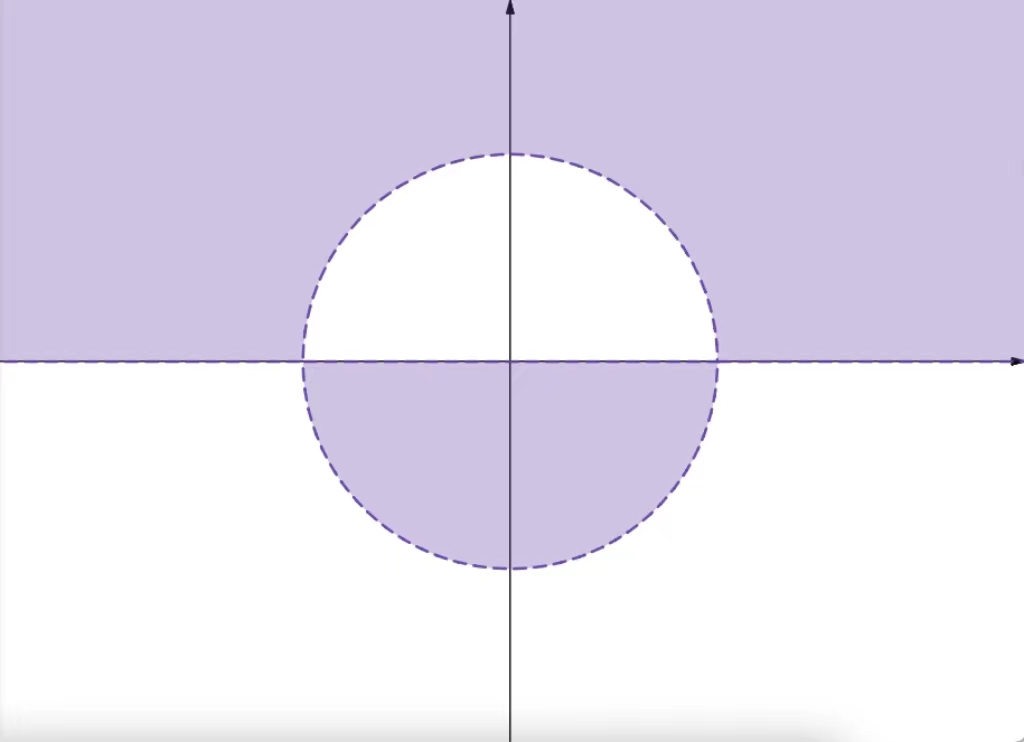}
		\end{minipage}
	}\qquad\,\,\,\,\,
	\subfigure[  $e^{2it\theta(\omega z)}\to0$ as $t\to \infty$ in  the purple region, $e^{-2it\theta(\omega z)}\to0$ as $t\to \infty$ in  the white region]{\label{c}
		\begin{minipage}[t]{0.25\linewidth}
			\centering
			\includegraphics[width=1.5in]{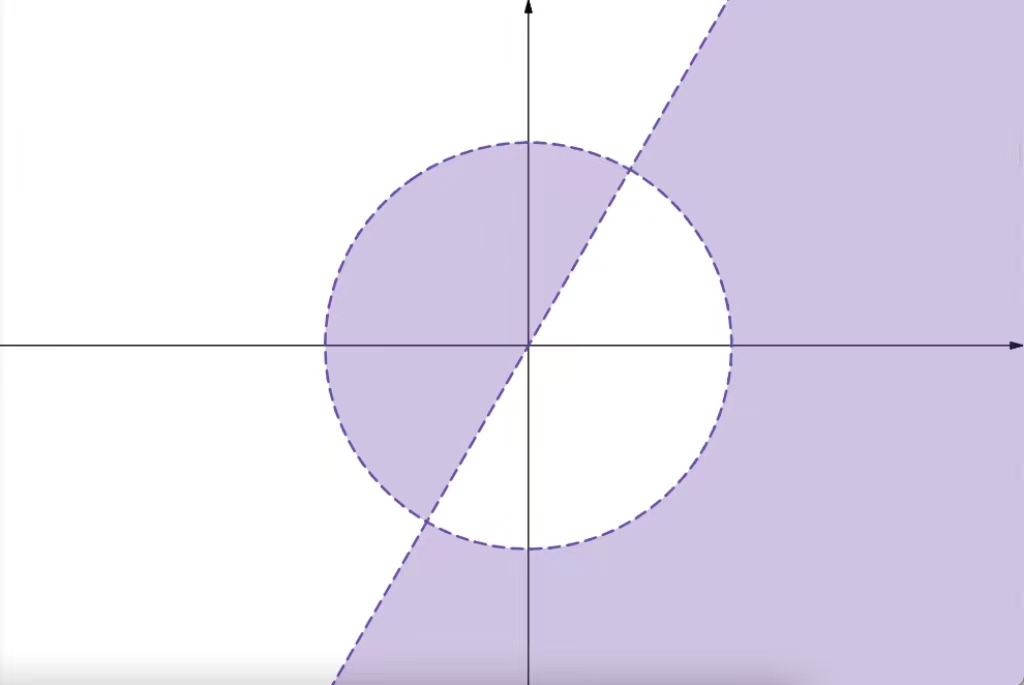}
		\end{minipage}
	}\qquad\,\,\,\,\,
	\subfigure[  $e^{2it\theta(\omega^2 z)}\to0$ as $t\to \infty$ in  the purple region, $e^{-2it\theta(\omega^2 z)}\to0$ as $t\to \infty$ in  the white region]{\label{c}
		\begin{minipage}[t]{0.25\linewidth}
			\centering
			\includegraphics[width=1.5in]{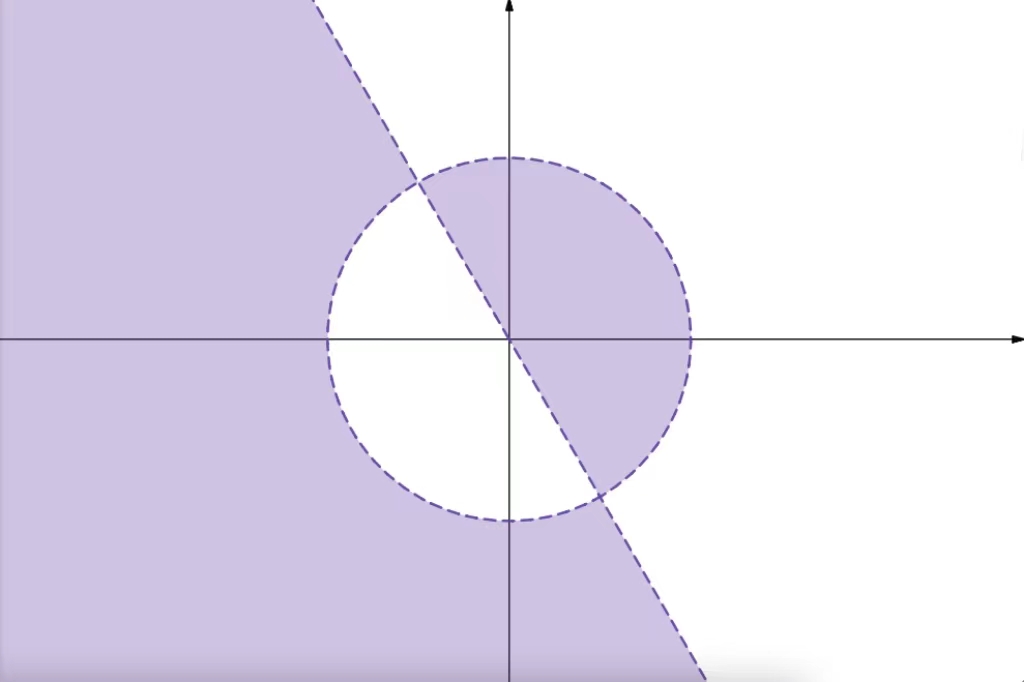}
		\end{minipage}
	}
	\caption{There are six  phase  points on the contour $\Sigma$  and signature tables   for the case  $\xi<0$}
	\label{sig2}
\end{figure}
\item[$\blacktriangleright$]   For the range $\xi>0$, there is no phase  point on $\Sigma$, see Figure \ref{sig1}.
The soliton resolution for the OV equation  in this case will be presented in Section \ref{sec44}.		
	\end{enumerate}

\subsection{Conjugation}\label{sec2.8}
Consider the residue  at $\xi_n$  in RH problem \ref{RH2.1}, in order
 to distinguish two types of zeros, we denote
 \begin{align*}
 	&{\Lambda}_1=\{n\in\mathcal{N}_1\,:\,{\rm Im}\,\theta(\xi_n)<0
 	\},\qquad
 	{\Lambda}_3=\{n\in\mathcal{N}_1\,:\,{\rm Im}\,\theta(\xi_n)>0\},\\
 	&{\Lambda}_2=\{n\in\mathcal{N}_2\,:\,{\rm Im}\,\theta(\omega \xi_n)>0
 	\},\quad\,\,\,	{\Lambda}_4=\{n\in\mathcal{N}_2\,:\,{\rm Im}\,\theta(\omega \xi_n)<0
 	\}
 	.
 \end{align*}
 where
 \begin{equation*}
 	{\mathcal{N}_1}=\{ 1,\dots,N_1\},\quad {\mathcal{N}_2}=\{N_1+1,\dots,N\}.
 \end{equation*}
 For  $n\in\Lambda^+={\Lambda}_1\cup {\Lambda}_2$,  it grows without bound as $t\to+\infty$, for $n\in\Lambda^-={\Lambda}_3\cup {\Lambda}_4$, it
approaches zero.

We introduce the notation on the contour $\mathbb{R}$
\begin{align*}
&	I_1=(-\infty,-\varkappa),\quad
	I_-=(-\varkappa,0),\quad
	I_+=(0,\varkappa),\quad I_2=(\varkappa,+\infty),\\
&{I}= \varnothing,\ \ {\rm for} \  \xi>0; \ \
		I=I_1\cup	I_2,\ \ {\rm for} \ \xi<0,
\end{align*}
and denote $
 {I}^{\omega}=\{\,\omega z\,:\,z\in{I}\,\},\quad
{I}^{\omega^2}=\{\,\omega^2 z\,:\,z\in{I}\,\}.
$

Consider the following scalar RH problem

\begin{problem}\label{RH22}
	Find a scalar function $\delta(z)$  analytical  for $z\in\mathbb{C}\setminus{I}$ with the following properties:
	\begin{enumerate}
		\item $\delta(z)\to1$ as $z\to\infty$.
		\item $\delta(z)$ has continuous boundary values $\delta_{\pm}(z)$ for $z\in\mathbb{C}\setminus{I}$
		and obey the jump relation
		\begin{equation*}
			\delta_+(z)=\begin{cases}
				(1-|r(z)|^2) \delta_-(z),\qquad\, z\in{I},\\
				\delta_-(z),\qquad\quad z\in\mathbb{R}\setminus{I}.
			\end{cases}
		\end{equation*}
	\end{enumerate}
\end{problem}
By using the  Plemelij formula, it can shown that for $r\in H^1(\mathbb{R})$,  the RH problem \ref{RH22} has a unique solution
		\begin{equation}\nonumber
			\delta(z)=
\exp\left\{ i\int_{{I}}\frac{\nu(s)}{s-z}ds\right\},
		\end{equation}	
		where   $\nu(z)$ is given by
	\begin{equation}\label{eq:63}
		\nu(z)=-\frac{1}{2\pi}\log(1-|r(z)|^2).
	\end{equation}

\begin{figure}[htp]
	\subfigure[ $e^{2it\theta(z)}\to0$ as $t\to \infty$ in  the purple region, $e^{-2it\theta(z)}\to0$ as $t\to \infty$ in  the white region]{\label{a}
		\begin{minipage}[t]{0.25\linewidth}
			\centering
			\includegraphics[width=1.5in]{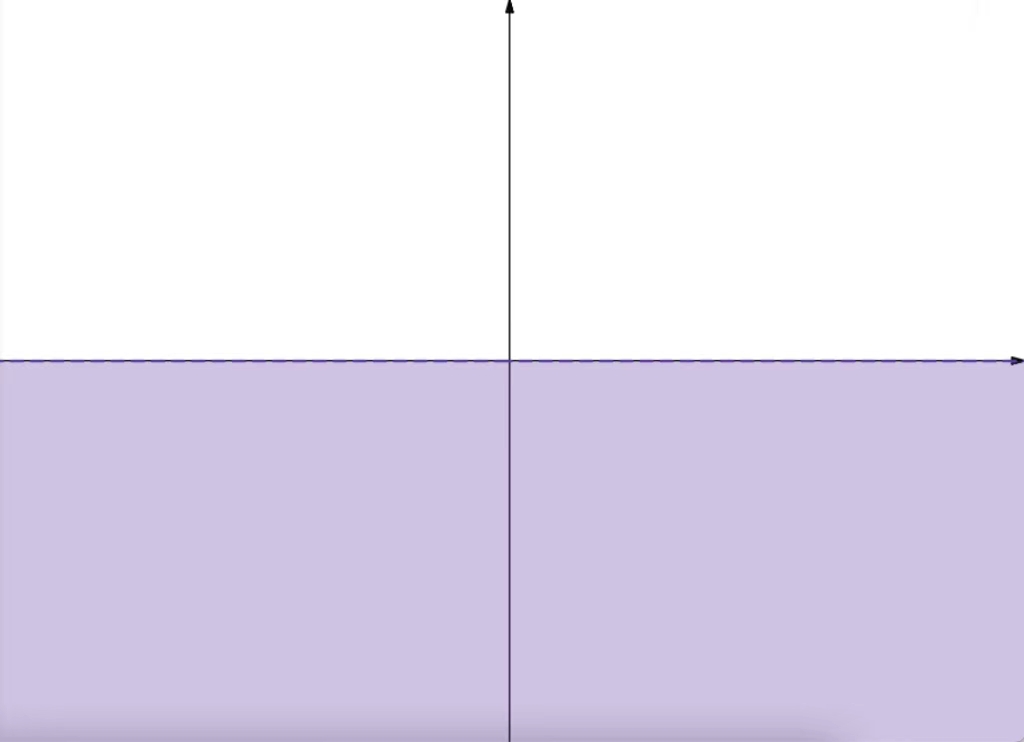}
		\end{minipage}
	}\qquad\,\,\,\,\,
	\subfigure[  $e^{2it\theta(\omega z)}\to0$ as $t\to \infty$ in  the purple region, $e^{-2it\theta(\omega z)}\to0$ as $t\to \infty$in  the white region]{\label{c}
		\begin{minipage}[t]{0.25\linewidth}
			\centering
			\includegraphics[width=1.5in]{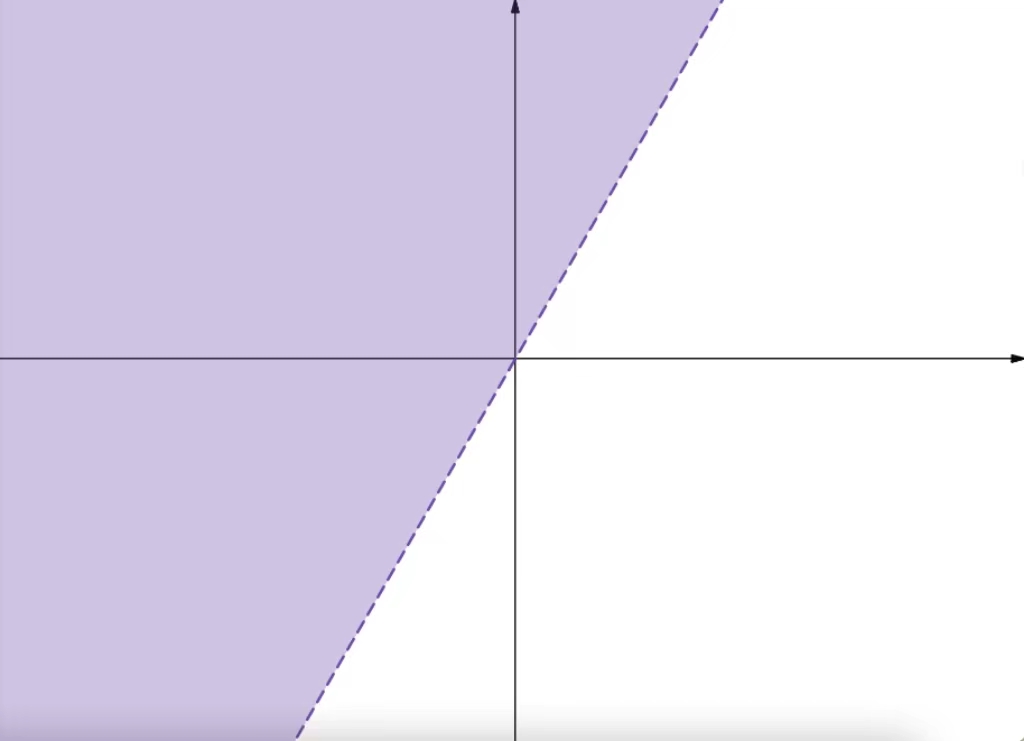}
		\end{minipage}
	}\qquad\,\,\,\,\,
	\subfigure[  $e^{2it\theta(\omega^2 z)}\to0$ as $t\to+\infty$ in  the purple  region, $e^{-2it\theta(\omega^2 z)}\to0$ as $t\to \infty$ in  the white region]{\label{c}
		\begin{minipage}[t]{0.25\linewidth}
			\centering
			\includegraphics[width=1.5in]{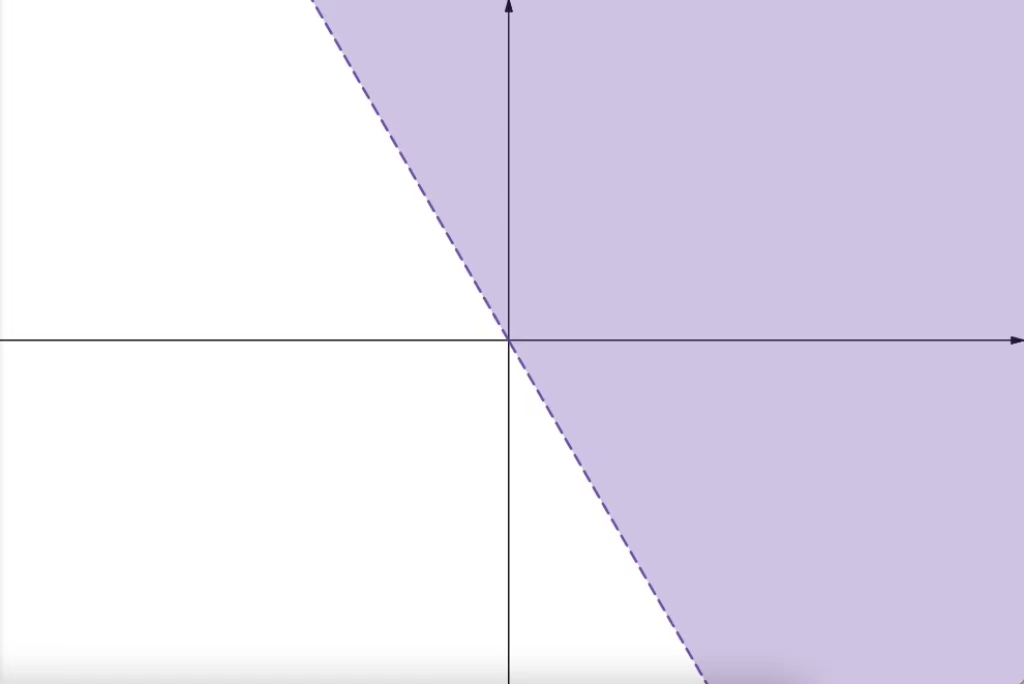}
		\end{minipage}
	}
	\caption{There is no  phase  point  and signature tables  for the case   $\xi>0$}
	\label{sig1}
\end{figure}

Further  we define
\begin{align}
&T(z)=\prod_{n\in{\Lambda}_1}\Big(\frac{z-\bar \xi_n}{z-\xi_n}\Big)\prod_{n\in{\Lambda}_2}\Big(\frac{z-\omega^2\bar\xi_n}{z- \omega\xi_n}\Big)
\delta(z),\label{eq:240}\\
	&{F}_1(z)=\frac{T(z)}{T(\omega^2z)},\,\,{F}_2(z)=F_1(\omega z),\,\,{F}_3(z)=F_1(\omega^2 z),\label{eq:243}\\
	&{F}_{ij}(z)=\frac{{F}_i(z)}{{F}_j(z)},\quad i,j=1,2,3,\label{eq:2246}
\end{align}

\begin{proposition}\label{pro24}  The functions defined by (\ref{eq:243})-(\ref{eq:2246})  admit the   properties
	
	\begin{enumerate}
		  	\item[{\rm (1)}]
${F}_1(z)$ is meromorphic in $\mathbb{C}\setminus{I}$. For each $n\in{\Lambda}^-$,  ${F}_1(z)$  has  simple poles at $\omega^2\xi_n,\omega^2\bar \xi_n$ and  simple zeros at $\omega\bar \xi_n,\xi_n$;  For each $n\in{\Lambda}^+$,  ${F}_1(z)$  has  simple poles at  $\omega\xi_n, \bar\xi_n$   and  simple zeros at $ \omega^2\bar\xi_n,\omega^2\xi_n$.
	
		\item[{\rm (2)}]	The function ${F}_1(z)$ fulfills the jump relation:
		\begin{align*}
			&{F}_{1+}(z)=(1-|r(z)|^2){F}_{1-}(z),\qquad z\in{I},\\
			&{F}_{1+}(z)=(1-|r(\omega^2z)|^2){F}_{1-}(z),\quad z\in{I}^\omega.
		\end{align*}
			\item[{\rm (3)}]
		when $z\to0$ , we have
		$${F}_3(z)={F}_3^0+{F}_3^1 z+\mathcal{O}(z^2),$$
		where
		\begin{align}\label{rr32}
			&F_3^0= \prod_{n\in{\Lambda}^-}\Big(\frac{\omega^2\bar \xi_n\cdot\omega \xi_n}{\xi_n\cdot\bar \xi_n}\Big)\prod_{n\in{\Lambda}^+}\Big(\frac{\xi_n\cdot\bar\xi_n}{\omega\bar\xi_n\cdot\omega^2\xi_n}\Big),\\\nonumber
			&{F}_3^1=\frac{\sqrt{3}F_3^0}{\pi}\int_\varkappa^\infty\frac{\log(1-|r(s)|^2)}{s^2}ds.
		\end{align}	
	\item[{\rm (4)}] Along any ray $l=\varkappa_{0,j}+e^{i\phi}\mathbb{R}^+, j=1,2$ with $|\phi|<\pi/6$,
		\begin{equation}\label{eq:2399}
			|{F}_{12}(z)-{F}^0_{12}(\varkappa_{0,j})(z-\varkappa_{0,j})^{i\nu(\varkappa_{0,j})}|\lesssim\|r\|_{H^1(\mathbb{R})}|z-\varkappa_{0,j}|^{1/2},
		\end{equation}		
		where ${F}^0_{12}(
	\varkappa_{0,j})$ is the complex unit:
		\begin{equation}
			{F}^0_{12}(
			\varkappa_{0,j}           )=f_1f_2e^{2i\beta(\varkappa_{0,j},\varkappa_{0,j})},
		\end{equation}
		and
		\begin{align*}
			&f_1=\prod_{n\in{\Lambda}_1}
			\Big(\frac{\varkappa_{0,j}-\xi_n}{\varkappa_{0,j}-\bar \xi_n}\Big)^2\frac{
				\varkappa_{0,j}-\omega\bar \xi_n}{\varkappa_{0,j} -\omega \xi_n}\cdot\frac{\varkappa_{0,j}-\omega^2\bar\xi_n}{
				\varkappa_{0,j}-\omega^2\xi_n},\\
			&f_2=\prod_{n\in{\Lambda}_2}	\Big(\frac{\varkappa_{0,j}  -\omega\xi_n}{\varkappa_{0,j}  -\omega^2\bar\xi_n}\Big)^2\frac{
			\varkappa_{0,j}  -\bar\xi_n}{\varkappa_{0,j} -\omega^2\xi_n}\cdot\frac{
			\varkappa_{0,j}   -\omega\bar\xi_n}{
			\varkappa_{0,j}   -\xi_n},
		\end{align*}
		and
		\begin{equation*}
			\beta(\varkappa_{0,j},z)=\int_{{I}}\frac{\nu(s)-\Xi(s)\nu(
			\varkappa_{0,j}  )}{s-z}ds-\nu(\varkappa_{0,j}
			)\log(z-\varkappa_{0,j} -1),
		\end{equation*}	
		where $\Xi(s)$ is characteristic functions on the interval $\varkappa-1< s<\varkappa.$
	\end{enumerate}
\end{proposition}
\begin{proof}
	Properties (1)-(3) can be obtain by simple calculation from the definition of equation \eqref{eq:243}-\eqref{eq:2246} and Plemelj formula. And for (4), analogously to \cite{MR18}. Via the fact that
	$$(z
-\varkappa_{0,j})^{i\nu(
		-\varkappa_{0,j}
		)}\leq e^{-\pi\nu(
		-\varkappa_{0,j}
		)}=\sqrt{1+|r(
	-\varkappa_{0,j}
		)|^2},$$
	an using Lemma 23.3 of \cite{BR88}
	$$|\beta(\varkappa_{0,j},z)-\beta(\varkappa_{0,j},
	\varkappa_{0,j})|\lesssim\|r\|_{H^1(\mathbb{R})}|z-\varkappa_{0,j}|^{1/2}.$$
	Then the results follow promptly.
\end{proof}

Let   $Y(z) ={\rm diag} \{ {F}_1(z),  {F}_2(z),  {F}_3(z)\}$,
then    function
\begin{equation}\label{eq:2.32}
	{M}^{(1)} (z) =M( z) Y(z)^{-1}
\end{equation}
 satisfies the following RH problem.

\begin{problem}\label{RH3.2}
	Find a analytical  function ${M}^{(1)} ( z): \mathbb{C}\setminus(\Sigma\cup\mathcal{Z})\to SL_3(\mathbb{C})$ with the following properties
	
	$\blacktriangleright$\emph{{ The normalization condition:}}
	$${M}^{(1)} (z) =I+\mathcal{O}(z^{-1})\quad z\to\infty.$$
	
	$\blacktriangleright$\emph{{ The symmetry proposition:}}	
	$${M}^{(1)} ( z)=\Gamma_1\overline{{M}^{(1)} ( \bar z)}\Gamma_1=\Gamma_2\overline{{M}^{(1)} ( \omega^2\bar z)}\Gamma_2=\Gamma_3\overline{{M}^{(1)} ( \omega \bar z)}\Gamma_3.$$

	$\blacktriangleright$\  ${M}^{(1)}(z)$   satisfies  the  jump relation:
	\begin{align*}
		 {M}^{(1)}_+ ( z)  = {M}^{(1)}_- ( z) \begin{cases}
			J^{(1)} ( z) ,\,\quad\hspace{1.4cm} z\in\mathbb{R},\\
			\Gamma_4^{-1}J^{(1)} (\omega^2z)\Gamma_4 ,\,\quad\, z\in\omega\mathbb{R},\\
			\Gamma_4J^{(1)} ( \omega z)\Gamma_4^{-1} ,\,\quad\,\hspace{0.2cm} z\in\omega^2\mathbb{R},
		\end{cases}
	\end{align*}
	where	
	\begin{align}\label{eq:2.48}
		&J^{(1)} ( z)\\\nonumber
		&=\begin{pmatrix}
			1&0&0\\
			-r(z){F}_{12}(z)e^{2it\theta(z)}&1&0\\
			0&0&1
		\end{pmatrix}\begin{pmatrix}
			1&\bar{r}(z){F}_{12}(z)^{-1}e^{-2it\theta(z)}&0\\
			0&1&0\\
			0&0&1
		\end{pmatrix},\,\,\qquad\quad\,\, z\in\mathbb{R}\setminus{I},\\\nonumber
		&=\begin{pmatrix}
			1&\frac{\bar r(z)}{1-|r(z)|^2}{F}_{12}(z)^{-1}e^{-2it\theta(z)}&0\\
			0&1&0\\
			0&0&1
		\end{pmatrix}\begin{pmatrix}
			1&0&0\\
			-\frac{r(z)}{1-|r(z)|^2}{F}_{12}(z)e^{2it\theta(z)}&1&0\\
			0&0&1
		\end{pmatrix},\,\,\,\, z\in{I}.
	\end{align}

	$\blacktriangleright$\
	${M}^{(1)} (z)$ has simple poles   $\xi_n\in\mathcal{Z}^+, \ \bar\xi_n\in \bar{\mathcal{Z}}^+, \ n=1,\dots 3N$  at which
	\begin{align}\nonumber
	&	\underset{z=\xi_n\phantom{+2n}}{\rm\,Res}\,M^{(1)}( z) =\underset{z\to \xi_n}{\lim}M^{(1)}(z){V}^{(1)}(z), \\\nonumber
	&\underset{z= \xi_{n+N}\phantom{n}
 }{\rm\,Res}\,M^{(1)}( z) =\underset{z\to \xi_{n+N}}{\lim}M^{(1)}(z)\Gamma_3\bar{V}^{(1)}(z)\Gamma_3, \\	\nonumber
	&\underset{z=\xi_{n+2N}}{\rm\,Res}\,M^{(1)}( z) =\underset{z\to \xi_{n+2N}}{\lim}M^{(1)}(z)\Gamma_1{V}^{(1)}(z)\Gamma_1, \\\nonumber
	&\underset{z=\bar\xi_n\phantom{+2n}}{\rm\,Res}\,M^{(1)}( z) =\underset{z\to\bar\xi_n\phantom{+2}}{\lim}M^{(1)}( z)\Gamma_4^{-1}\bar{{V}}^{(1)}(z)\Gamma_4,\\\nonumber
	&\underset{z=\bar\xi_{n+N}\phantom{n}
 }{\rm\,Res}\,M^{(1)}( z) =\underset{z\to\bar\xi_{n+N}}{\lim}M^{(1)}(z)\Gamma_4^{-1}\Gamma_3\bar{V}^{(1)}(z)\Gamma_3\Gamma_4,\\\nonumber
	&\underset{z=\bar\xi_{n+2N}}{\rm\,Res}\,M^{(1)}( z) =\underset{z\to \bar\xi_{n+2N}}{\lim}M^{(1)}(z)\Gamma_4^{-1}\Gamma_1\bar {V}^{(1)}(z)\Gamma_1\Gamma_4,
\end{align}
	where ${V}^{(1)}(\xi_n)$ is the nilpotent matrix. For each $n=1,\dots,N_1$, we have
	\begin{align*}
		&	{V}^{(1)}(\xi_n)=
			\begin{pmatrix}
				0&0&0\\
				\varrho^0_n(\xi_n)&0&0\\
				0&0&0
			\end{pmatrix},\,n\in\Lambda^-,\ \
			{V}^{(1)}(\xi_n)=\begin{pmatrix}
				0&\varrho^1_n(\xi_n)&0\\
				0&0&0\\
					0&0&0
			\end{pmatrix},\,\,n\in\Lambda^+,
	\end{align*}
	with
	\begin{align*}
		&\varrho^0_n(\xi_n)=c_n e^{-2it\theta(\xi_n)}{F}_{12}(\xi_n),\,\,\varrho^1_n(\xi_n)=(c_n e^{-2it\theta(\xi_n)})^{-1}(1/{F}_1)'(\xi_n)^{-1}(1/{F}_2)'(\xi_n),
	\end{align*}
	while for each $ n=N_1+1,\dots,N$, we acquire		
	\begin{align*}
		&{V}^{(1)}(\xi_n)=
			\begin{pmatrix}
				0&0&0\\
				0&0&\vartheta^0_n(\xi_n)\\
				0&0&0
			\end{pmatrix},\,n\in\Lambda^-,\ \ {V}^{(1)}(\xi_n)=\begin{pmatrix}
				0&0&0\\
				0&0&0\\
				0&\vartheta^1_n(\xi_n)&0
			\end{pmatrix},\ n\in\Lambda^+,
	\end{align*}
	where
	\begin{align*}
		&\vartheta^0_n(\xi_n)=c_n e^{-2it\theta(\omega \xi_n)}{F}_{23}(\xi_n),\,\,\vartheta^1_n(\xi_n)=(c_n e^{-2it\theta(\omega\xi_n)})^{-1}(1/{F}_2)'(\xi_n)^{-1}(1/{F}_3)'(\xi_n).
	\end{align*}
	
\end{problem}

\section{Soliton Resolution in the   Region  I: $x/t<0$}\label{sec33}

In this section, we give the soliton resolution for the   OV equation (\ref{eq:1.1})  in  the asymptotic    region $x/t<0$.

\subsection{Opening $\bar\partial$-lenses}\label{sec31}

  We  construct appropriate continuous extension functions,
which  deform  the oscillatory jump along the $\Sigma$ onto new contours
$$\Sigma^{(2)}=L\cup\omega L\cup\omega^2L,\quad L=\cup_{j=1}^4l^0_j  \cup l^0_{5\pm}\cup l^0_{6\pm},$$
along which the jumps are decaying, see Figure\,\ref{F4.1}. We designate the areas between $l^0_j$ and ${\rm Re}\,z$  in a clockwise manner as $\Omega^0_j$ for $j=1,\ldots,4$, and for $j=5,6$, we define $\Omega^0_{5_{\mp}}, \Omega^0_{6_{\pm}}$. Subsequently, we proceed to define $\Omega^1_j$ and $\Omega^2_j$ using a same way.

\begin{figure}[htp]
{
		\begin{minipage}{15cm}\centering                               			
			\includegraphics[scale=0.4]{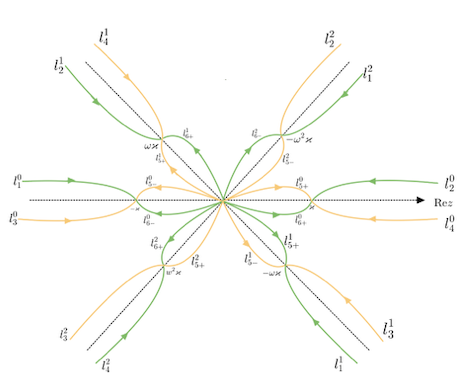}
	\end{minipage}}
	\caption{The contour $\Sigma^{(2)}$ for the RH problem \ref{RH4.1}.}
\label{F4.1}
\end{figure}

Denote
\begin{equation}\label{eq:311}
	\rho=\frac{1}{2} \{\underset{z\ne z'\in\mathcal{Z}}{\min}|z-z'|, \ \mbox{dist}(\mathcal{Z},\mathbb{R}), \ \mbox{dist}(\mathcal{Z}, \omega\mathbb{R}), \ \mbox{dist}(\mathcal{Z},\omega^2\mathbb{R})\},
\end{equation}
which is sufficiently small to ensure that the disks of support intersect neither each other nor the contour $\Sigma$.
Denote  $\check\chi(z)=1-\chi_{\mathcal{Z}}(z)$, where
$\chi_{\mathcal{Z}}\in C_0^\infty(\mathbb{C},[0,1])$  is  supported near the discrete spectrum such that
\begin{equation}\label{eq:3310}
	\chi_\mathcal{Z}(z)=\begin{cases}
		1\quad \mbox{dist}(z,\mathcal{Z})<\rho/3,\\
		0\quad \mbox{dist}(z,\mathcal{Z})> 2\rho/3.
	\end{cases}
\end{equation}
Further  denote
\begin{align}\label{eq:333}
	{G}_n(z)=\begin{cases}
		{F}_{12}(z),\quad n=0,\\
		{F}_{13}(z),\quad n=1,\\
		{F}_{32}(z),\quad n=2,
	\end{cases}
\end{align}
and
\begin{align*}
	&	p^0_1(z)=p^0_2(z)=\frac{r(z)}{1-|r(z)|^2},\quad p^0_{5-}(z)=p^0_{5+}(z)=-\bar r(z),\\
	&	p^0_3(z)=p^0_4(z)=\frac{\bar r(z)}{1-|r(z)|^2},\quad p^0_{6-}(z)=p^0_{6+}(z)=-r(z),
\end{align*}
for $j=1,2,3,4,5_\pm,6_\pm$, we have
\begin{equation}\label{eq:3334}
	p^1_j(z)=p^0_j(\omega^2z),\quad p^2_j(z)=p^0_j(\omega z).
\end{equation}

In order to deform the contour $\Sigma$ to the contour $\Sigma^{(2)}$,
we introduce  the following extension   functions.

\begin{lemma}\label{le4.1}
	It is possible to define functions $R^n_j:\bar\Omega^n_j\to\mathbb{C},\,n=0,1,2,j=1,2,3,4,5_\pm,6_\pm$ with boundary values satisfying
	\begin{align}\nonumber
		&R^n_{j}(z)=\begin{cases}
			p_j^n( z)G_n(z),\,\,\qquad\qquad\qquad\qquad\qquad\qquad\qquad\,\,\,\,\,\,\,\,\,\hspace{0.7cm}j=1,2,z\in\omega^nI_{j},\\[2pt]
			p_j^{n}((-1)^j\varkappa)G_{n}((-1)^j\omega^n\varkappa)(z-(-1)^{j}\omega^n\varkappa)^{2i\nu((-1)^j\omega^n\varkappa)}
			\check\chi(z),\hspace{0.4cm}z\in\omega^n l^n_j,
		\end{cases}\\[4pt] \nonumber
		&R^n_{j}(z)=\begin{cases}
			p_j^n( z){G}_{n}(z)^{-1},\,\,\,\,\,\,\,\quad\quad\qquad\qquad\qquad\qquad\,\,\,\,\hspace{2.2cm}j=3,4,z\in\omega^nI_{j},\\[2pt]
			p_j^{n}((-1)^j\varkappa){G}_{n}((-1)^j\omega^n\varkappa)^{-1}(z-(-1)^{j}\omega^n\varkappa)^{-2i\nu((-1)^j\omega^n\varkappa)}
			\check\chi(z),\hspace{0.1cm}z\in\omega^n l^n_j,
		\end{cases}\\[4pt] \nonumber
		&R^n_{5_\pm}(z)=\begin{cases}
			p^{n}_{5_\pm}(z)G_{n}(z)^{-1},\qquad\qquad\qquad\qquad\qquad\hspace{1.9cm} z\in\omega^nI_{\pm},\\[2pt]
			p^{n}_{5_\pm}(\pm\varkappa)G_{n}(\pm\omega^n\varkappa)^{-1}(z\mp\omega^n\varkappa)^{-2i\nu(\pm\omega^n\varkappa)}\check\chi(z),\hspace{0.3cm} z\in\omega^nl^n_{5_\pm},\\
		\end{cases}\\[4pt]\label{eq:3.9}
		&R^n_{6_\pm}(z)=\begin{cases}
			p^{n}_{6_\pm}(z)G_{n}(z), \hspace{6.3cm}z\in\omega^nI_{\pm},\\[2pt]
			p^{n}_{6_\pm}(\varkappa)G_{n}(\pm\omega^n\varkappa)(z\mp\omega^n\varkappa)^{2i\nu(\pm\omega^n\varkappa)}\check\chi(z)\qquad\,\,\,\,\,\,\,\,\,\,\,\,\,\,   z\in\omega^nl^n_{6_\pm},
		\end{cases}
	\end{align}
	such  that for $j=1,\dots4$, we get
	\begin{align}\label{eq:4.2}
		&	|\bar\partial R^n_j(z)|\lesssim|p_j'({\rm Re}(z)|+|z\pm\omega^n\varkappa|^{-1/2}+\bar\partial(\chi_{\mathcal{Z}}),
	\end{align}
	for $j=5_\pm,6_\pm$, we have
	\begin{align}\label{eq:3133}
		&	|\bar\partial R^n_j(z)|\lesssim |\varphi({\rm Re}\,z)|+|p_j'({\rm Re}(z)|+|z\pm\omega^n\varkappa|^{-1/2}+\bar\partial(\chi_{\mathcal{Z}}).
	\end{align}
where  $\varphi\in C_0^\infty(\mathbb{R}, [0,1])$  is a fixed cutoff function near small support near $z=0$.
	Moreover, when $z\to0$,
	\begin{equation}\label{eq:316}
		|\bar\partial R^n_j(z)|\lesssim |z|^2,\quad n=0,1,2.
	\end{equation}
	
\end{lemma}

\begin{proof}
We give the details firstly for
	$z =-\varkappa+ se^{i\alpha}\in\Omega^0_{5-}
	$. The other cases are easily inferred.
	Define the function $R^0_{5-}(z)$ as follows:
	\begin{equation*}
		R^0_{5-}(z)= R^{01}_{5-}(z)+ R^{02}_{5-}(z),
	\end{equation*}	
	where
	\begin{align*}
		&R^{01}_{5-}(z)=(\varphi({\rm Re}\,z)-1)\bar r({\rm Re}\,z){G}_{0}(z)^{-1}\cos^2 \alpha +f_{5-}(z)(1-\cos^2  \alpha ),\\
		&R^{02}_{5-}(z)=g_{5-}( {\rm Re}\,z)\cos^2 \alpha +\frac{i}{2}s^{-i\alpha}\sin(2\alpha)\cos(2\alpha)\varphi(\alpha)g'_{5-}({\rm Re}\,z)+\sin^2 \alpha \\
		&\cdot \varphi(\alpha)g_{5-}({\rm Re}\,z)-\frac{1}{4}se^{-i\alpha}\sin^2 \alpha g'_{5-}({\rm Re}\,z)-\frac{1}{8}s^2e^{-2i\alpha}\sin^2 \alpha g''_{5-}({\rm Re}\,z),
	\end{align*}
	where	we take the cut-off function $\varphi(z)\in\mathbb{C}_0^\infty(\mathbb{R}[0,1])$ satisfying
	\begin{equation}
		\varphi(z)=\begin{cases}
			1,\quad |z|<\min\{\rho,1\}/8,\\
			0,\quad |z|> \min\{\rho,1\}/4,
		\end{cases}
	\end{equation}
	where $\rho$ defined as \eqref{eq:311} and
	\begin{align*}
		&g_{5-}({\rm Re} \,z)=-\varphi({\rm Re}\,z)\bar r({\rm Re}\,z)G_{0}(z)^{-1},\\
		&f_{5-}(z)=	-\bar r(\varkappa)G_0(\varkappa)^{-1}(z\varkappa)^{-2i\nu(\varkappa)}(1-\chi_{\mathcal{Z}}(z)).
	\end{align*}

	We first deal with $R^{01}_{5-}(z)$, by using $\bar\partial=\frac{1}{2}e^{i\alpha}\big(\frac{\partial}{\partial s}+\frac{i}{s}\frac{\partial}{\partial\alpha}\big)$. Given that $\bar r'({\rm Re}\,z),s^{-1}$ is bounded in the domain of $\varphi({\rm Re}\,z)$.
	Due to the bound  of function $G_0({\rm Re}\,z)^{-1}$ then we get
\begin{align}
|	\bar\partial R^{01}_{5-}(z)|\lesssim \varphi({\rm Re} \,z)+|r'({\rm Re}\,z)|+|z-\varkappa|^{-1/2}.\label{po1}
\end{align}
	Similarly, as for $	\bar\partial R^{02}_{5-}(z)$, we have
	$$|	\bar\partial R^{02}_{5-}(z)|\lesssim \varphi({\rm Re}\,z).$$
	
	As $z\to0$, we have $\alpha\to0$ and within a small neighborhood of $0$, $\varphi(\alpha)=1,\varphi'(\alpha)=0$, thus
	\begin{align}
		&	|\bar\partial R^{02}_{5-}(z)|\lesssim|\sin^2\alpha|+|(g_{5-})'({\rm Re}\,z)\varphi'(\alpha)\sin\alpha|\lesssim|z|^2, \label{po2}
	\end{align}
which together with (\ref{po1}) gives the estimate (\ref{eq:3133}) for the case $z  \in\Omega^0_{5-}
	$.	
\end{proof}

We now define continuous extension function
\begin{align}\nonumber
	\mathcal{R}^{(2)}=\begin{cases}
		\begin{pmatrix}
			1&0&0\\
			R^0_je^{2it\theta(z)}&1&0\\
			0&0&1
		\end{pmatrix},\,\,\,\,\,\,z\in\Omega^r_{j},\,\,\,\,\,
		\begin{pmatrix}
			1&R^0_{j}e^{-2it\theta(z)}&0\\
			0&1&0\\
			0&0&1
		\end{pmatrix},\,\,\,\,z\in\Omega^e_j,\\[15pt]
		\begin{pmatrix}
			1&0&0\\
			0&1&0\\
			R^1_je^{2it\theta(\omega^2z)}&1&0\\
		\end{pmatrix},\,z\in\omega\Omega^r_{j},
		\begin{pmatrix}
			1&0&R^1_{j}e^{-2it\theta(\omega^2z)}\\
			0&1&0\\
			0&0&1
		\end{pmatrix},\,z\in\omega\Omega^e_j,\\[15pt]
		\begin{pmatrix}
			1&0&0\\
			0&1&	R^2_je^{2it\theta(\omega z)}\\
			0&0&1
		\end{pmatrix},\,\,\,\,z\in\omega^2\Omega^r_{j},
		\begin{pmatrix}
			1&0&0\\
			0&1&0\\
			0&R^2_{j}e^{-2it\theta(\omega z)}&1
		\end{pmatrix},\,z\in\omega^2\Omega^e_j,\\[15pt]
		I\qquad\quad\hspace{2.9cm} z\in\,elsewhere.
	\end{cases}
\end{align}
where we denote $\Omega^{r}_{j}=\Omega^0_{j}|_{j=1}^2\cup\Omega^0_{6\pm}$ and $\Omega^{e}_j=\Omega^0_j|_{j=3}^4\cup\Omega^0_{5\pm}$.
Then  the  new known function
\begin{equation}
	{M}^{(2)} (z)={M}^{(1)} (z)\mathcal{R}^{(2)}(z). \label{eq:312}
\end{equation}
satisfies the following  mixed $\bar\partial$-RH problem.

\begin{problem}\label{RH4.1}
	Find a continuous with sectionally continuous first partial derivatives  function ${M}^{(2)} ( z):\mathbb{C}\setminus(\Sigma^{(2)}\cup\mathcal{Z})\to$SL$_3(\mathbb{C})$ with the following properties:
	
	$\blacktriangleright$\emph{{The normalization condition:}}
	$${M}^{(2)} (z)=I+\mathcal{O}(z^{-1})\quad z\to\infty.$$
	
	$\blacktriangleright$\emph{{The symmetry proposition:}}	
	$${M}^{(2)} ( z)=\Gamma_1\overline{{M}^{(2)} ( \bar z)}\Gamma_1=\Gamma_2\overline{{M}^{(2)} ( \omega^2\bar z)}\Gamma_2=\Gamma_3\overline{{M}^{(2)} ( \omega \bar z)}\Gamma_3.$$
	
	$\blacktriangleright$
	${M}^{(2)} ( z)$  satisfies  the  jump relation:
	\begin{align}\label{eq:44}
		{M}_+ ^{(2)} ( z)={M}_-^{(2)} ( z)\begin{cases}
			J^{(2)}( z),\,\quad\,\,\hspace{1.3cm} z\in L,\\
			\Gamma_4^{-1}J^{(2)}( \omega^2z)\Gamma_4,\,\quad\, z\in\omega L,\\
			\Gamma_4J^{(2)}( \omega z)\Gamma_4^{-1} ,\,\quad\,\,\, z\in\omega^2L,
		\end{cases}
	\end{align}
	where
		\begin{align*}
		&J^{(2)} ( z)=\begin{cases}
			\begin{pmatrix}
				1&0&0\\
				R^0_j(z)e^{2it\theta(z)}&1&0\\
				0&0&1
			\end{pmatrix},\, z\in l^0_{j=1,6+},\,\,\,\begin{pmatrix}
				1&R^0_{j}(z)e^{-2it\theta(z)}&0\\
				0&1&0\\
				0&0&1
			\end{pmatrix}^{-1},\, z\in l^0_{j=4,5-}, \\[15pt]
		\begin{pmatrix}
				1&R^0_{j}(z)e^{-2it\theta(z)}&0\\
				0&1&0\\
				0&0&1
			\end{pmatrix},\, z\in l^0_{j=3,5+},\,	\begin{pmatrix}
			1&0&0\\
			R^0_j(z)e^{2it\theta(z)}&1&0\\
			0&0&1
			\end{pmatrix}^{-1},\, z\in l^0_{j=2,6-},
		\end{cases}
	\end{align*}
	
	$\blacktriangleright$
	${M}^{(2)} ( z)$ satisfies  the same  residue conditions with ${M}^{(1)}( z)$.

	$\blacktriangleright$	For each $z\in\mathbb{C}\setminus(\Sigma^{(2)}\cup\mathcal{Z})$, we have
	  $\bar\partial$-equation
	\begin{align}\label{eq:4.4}
		\bar\partial{M}^{(2)} ( z) ={M}^{(2)} ( z) \bar\partial \mathcal{R}^{(2)}(z).
	\end{align}
\end{problem}

We further decompose the mixed $M^{(2)}( z)$ into the form
\begin{equation}\label{eq:3.177}
	{M}^{(2)}( z)=	M^{(3)}( z)M^{rhp} ( z),
\end{equation}
where    ${M}^{rhp}(z)$ is the solution of the pure RH problem \ref{RH3.4} and
 ${M}^{(3)}(z)$ is a continuously   function that satisfies the  $\bar\partial$-RH problem \ref{RH5.1}.

\subsection{Analysis on the pure  RH problem}\label{sec32}

We now consider the following pure  RH problem for    $M^{rhp}(z)$.
\begin{problem}\label{RH3.4}
	Find an analytic function $M^{rhp}(z):\mathbb{C}\setminus\Sigma^{(2)}\to SL_3(\mathbb{C})$ such that

	$\blacktriangleright$\emph{{The normalization condition:}}
	$$M^{rhp}(z)=I+\mathcal{O}(z^{-1})\quad z\to\infty.$$
	
	$\blacktriangleright$\emph{{The symmetry proposition:}}	
	$$M^{rhp}(z)=\Gamma_1\overline{M^{rhp}( \bar z)}\Gamma_1=\Gamma_2\overline{M^{rhp}( \omega^2\bar z)}\Gamma_2=\Gamma_3\overline{M^{rhp}( \omega \bar z)}\Gamma_3.$$

	$\blacktriangleright$
	$M^{rhp}(z)$  satisfies  the same  jump relations  and residue conditions with ${M}^{(2)}( z)$.

\end{problem}

Noticing that if
we define a open neighborhood
$$\mathcal{U}=  \bigcup_{n=0}^1 \bigcup_{j=0}^2 \mathcal{U}_{\varkappa_{nj}},     $$
where $\mathcal{U}_{\varkappa_{nj}}=\{z:|z-\varkappa_{nj}|<\rho/2\}, n=0,1; j=0,1,2$, then  within this vicinity,
the solution $M^{rhp}(z)$ is free of poles.   Therefore we  construct the solution $M^{rhp}(z)$ by the  decomposition
\begin{equation}\label{eq:320}
	M^{rhp}(z)=\begin{cases}
		\mathcal{E}(z){M}^{out}(z){M}^{loc}(z),\quad z\in \mathcal{U},\\
		\mathcal{E}(z){M}^{out}(z), \hspace{1.6cm} z \in  \mathbb{C}\setminus\mathcal{U},
	\end{cases}
\end{equation}
where each component  is constructed as follows:
\begin{enumerate}
	
	\item[$\bullet$]  ${M}^{out}(z)$ exactly solves the pure soliton   RH problem \ref{RH39} obtained by disregarding the jump matrix.

	\item[$\bullet$] ${M}^{loc}(z)$ matches parabolic cylinder  RH problem \ref{RH5.2}   in the vicinity of the saddle points $\varkappa_{nj} $.
	
	\item[$\bullet$]  $\mathcal{E}(z)$  is residual error which given by  the  small-norm RH problem \ref{RH5.7}.
\end{enumerate}

\subsubsection{A outer model about the soliton component}\label{3.3.2}
We  now solve the following the pure soliton   RH problem \ref{RH39}.
\begin{problem}\label{RH39}
	Find a 3$\times$3 analytic matrix-valued function ${M}^{out}(z)$\,:\,$\mathbb{C}\setminus\mathcal{Z}\to SL_3(\mathbb{C})$, with the following properties.

	$\blacktriangleright$\emph{{The normalization condition:}}
	$${M}^{out}(z)=I+\mathcal{O}(z^{-1})\quad z\to\infty.$$
	
	$\blacktriangleright$\emph{{The symmetry proposition:}}	
	$${M}^{out}(z)=\Gamma_1\overline{{M}^{out}(\bar z)}\Gamma_1=\Gamma_2\overline{{M}_{\dagger}(\omega^2\bar z)}\Gamma_2=\Gamma_3\overline{{M}^{out}(\omega \bar z)}\Gamma_3.$$

$\blacktriangleright$  ${M}^{out}(z)$ satisfies  the same  residue conditions with ${M}^{(2)}( z)$.
	
\end{problem}

For above RH problem, we can show that
\begin{proposition}   The RH problem \ref{RH39} admits a  unique solution   given by
	\begin{equation}\label{eq:33329}
		{M}^{out}(z)={M}^{sol}(z;  {\mathcal{D}}),
	\end{equation}
	where ${M}^{sol}(z;\hat{\mathcal{D}}) $ is the solution of RH problem \ref{RHB.1} corresponding to the reflectionless scattering data $\hat{\mathcal{D}}=\{r(z)\equiv0,(\xi_n,\hat C_n)_{n=1}^{6N}\}$.
	
	As $z\to0$, ${M}^{out}(z)$ admits the expansion
	\begin{equation}\label{eq:3.29}
		{M}^{out}(z)={M}_0^{out}+{M}_1^{out} z+\mathcal{O}(z^2).
	\end{equation}
	Then according to equation \eqref{eq:2.39}, we have
	\begin{align}
		u_{sol}(x,t;  {\mathcal{D}} ) = x_t,
	\end{align}
	where
	\begin{subequations}
		\begin{align*}
			& x(y,t)=y+ \frac{\sum_{j=1}^3[{M}_1^{out}]_{3j}}{\sum_{j=1}^3[{M}_0^{out}]_{3j}}.
		\end{align*}
		\begin{proof}
			   Existence and uniqueness of
       RH problem \ref{RH39}
      follows from Proposition \ref{prob1}, we give the reconstruction formula of $	
      {M}^{out}(z)$ from  \eqref{eq:2.39}.
		\end{proof}
	\end{subequations}
\end{proposition}

\subsubsection{A local model near the phase  points}\label{sec3.2.2}

 By utilizing the jump relations \eqref{eq:44}, the spectral bound \eqref{eq:311} and the phase functions $\theta(z)$, we can demonstrate that

\begin{equation}\label{eq:321}
	|J^{(2)} (z)-I|\lesssim\mathcal{O}(e^{-\frac{\sqrt{3}\rho}{2}t}), \ z\in\Sigma^{(2)}\setminus \mathcal{U},
\end{equation}
which is exponentially small outside of  $\mathcal{U}$. Based on this estimation, we can construct a model   in the region outside $\mathcal{U}$ that completely ignores the jumps.  Thus we  solve the following  local   RH problem   with contour $\Sigma^{loc}= \Sigma^{(2)}\cap \mathcal{U}$.

\begin{problem}\label{RH5.2}
	Find a 3$\times$3 analytic matrix-valued function ${M}^{loc} (z)$\,:\, $\mathbb{C}\setminus\Sigma^{loc}\to SL_3(\mathbb{C})$, with the following properties:

	$\blacktriangleright$\emph{{The normalization condition:}}
	$${M}^{loc} (z)=I+\mathcal{O}(z^{-1}),\quad z\to\infty.$$
	
	$\blacktriangleright$\emph{{The symmetry proposition:}}	
	$${M}^{loc} (z)=\Gamma_1\overline{{M}^{loc} ( \bar z)}\Gamma_1=\Gamma_2\overline{{M}^{loc} ( \omega^2\bar z)}\Gamma_2=\Gamma_3\overline{{M}^{loc} ( \omega \bar z)}\Gamma_3.$$	
	
	$\blacktriangleright$\emph{{The jump condition:}}
	${M}^{loc} (z)$ satisfies  the  jump relation:
	\begin{align}
		{M}^{loc}_+ (z) ={M}^{loc}_- (z)\begin{cases}
			J^{(2)} (z),\hspace{1.9cm} z\in L\cap \mathcal{U},\\
			\Gamma_4^{-1}J^{(2)} ( \omega^2z)\Gamma_4 ,\,\quad\, z\in   \omega L\cap \mathcal{U},\\
			\Gamma_4J^{(2)} ( \omega z)\Gamma_4^{-1} ,\,\quad\,\,\,\, z\in   \omega^2 L\cap \mathcal{U},
		\end{cases}	\end{align}
	where  $J^{(2)} (z)$ is given in \eqref{eq:44}.	
\end{problem}

In   Appendix \ref{secA},     we  take
\begin{align}
	&\zeta =\sqrt{c_{nj} t}(z-\varkappa_{nj}),\quad c_{nj}=\frac{2\sqrt{3}}{\varkappa_{nj}^3},\ n=0,1,2;j=0,1,\label{eq:6.18}\\
	&r_0=
	-\bar r(\varkappa_{nj})G_0(\varkappa_{nj})^{-1}e^{-2i\nu(\varkappa_{nj})\log(\sqrt{c_{nj}t})}
	\Big(\frac{8\sqrt{3}}{\varkappa_{nj}}t\Big)^{-i\nu}:=r_{\varkappa_{nj}},\label{r0s}
\end{align}
then the local RH problem \ref{RH5.2}  at each small cross centered at phase point $\varkappa_{nj}$
exactly matches parabolic cylinder model problem \ref{RH5.3}. Further we denote
\begin{align}
&M_1^{pc}(\varkappa_{0j}) =M_1^{pc}(r_0)\big|_{r_0=r_{\varkappa_{0j}}}, \ \ M_1^{pc}(\varkappa_{1j}) =\Gamma_4^{-1}  M_1^{pc}(\varkappa_{0j})\Gamma_4, \\
& M_1^{pc}(\varkappa_{2j}) =\Gamma_4  M_1^{pc}(\varkappa_{0j})\Gamma_4^{-1},\ \ j=1,2,
\end{align}
 where $M_1^{pc}(r_0) $ is given  by (\ref{pc1}) and $\varkappa_{0j} $ is defined by (\ref{r0s}) with $n=0$.   Then the solution  ${M}^{loc} (z)$ is consists of the sum of solutions
of the  RH problem  centered at $\varkappa_{nj}$, that is,
\begin{align}
	{M}^{loc}  (z)=  I&+	\frac{1}{\sqrt{t}}\sum_{n=0}^2  \sum_{j=1}^2  \frac{M_1^{pc}( \varkappa_{nj})}{\sqrt{c_{nj}}(  z-\varkappa_{nj})} +\mathcal{O}(\zeta^{-2}).\nonumber
\end{align}

\subsubsection{A small-norm RH problem}
In this section, we consider
the $\mathcal{E}(z)$ which is analytic in $\mathbb{C}\setminus\Sigma^{ \mathcal{E} }$, where
$$\Sigma^{ \mathcal{E} }=\partial \mathcal{U} \cup( \Sigma^{(2)}\setminus \mathcal{U}).$$
 It is straightforward to show that $\mathcal{E}(z)$ must satisfy the following RH problem.
\begin{problem}\label{RH5.7}
	Find a   matrix-valued function $\mathcal{E}(z)$ with the following properties:

	$\blacktriangleright$\emph{{The normalization condition:}}
	$$\mathcal{E}(z)=I+\mathcal{O}(z^{-1}),\quad z\to\infty.$$
	
	$\blacktriangleright$\
	$\mathcal{E}( z)$  satisfies  the  jump relation:
	\begin{equation}
		 \mathcal{E}_+ ( z) = \mathcal{E}_-( z) {V}^{\mathcal{E}}(z)
	\end{equation}
	where
	\begin{equation}\label{eq:329}
		{V}^{\mathcal{E}}(z)=\begin{cases}
			{M}^{out}( z)J^{(2)} ( z){M}^{out}( z)^{-1},\quad z\in\Sigma^{(2)}\setminus\mathcal{U},\\[1.5pt]
			{M}^{out}( z){M}^{loc} ( z){M}^{out}( z)^{-1},\quad\, z\in\partial\mathcal{U}.
		\end{cases}	
	\end{equation}
\end{problem}
Starting from  \eqref{eq:329} and \eqref{eq:321}, using the boundedness of ${M}^{out}(z)$ for $z\in\mathcal{U}$, one finds that
\begin{equation}\label{eq:332}
	|{V}^{\mathcal{E}}(z)-I|=
	\begin{cases}
		\mathcal{O}(e^{-\frac{\sqrt{3} \rho}{ 2} t }), \quad z\in\Sigma^{(2)}\setminus\mathcal{U},\\
		\mathcal{O}(t^{-1/2}),\quad\quad  z\in\partial\,\mathcal{U},
	\end{cases}
\end{equation}
and it follows that
\begin{equation}\label{eq?OE?332}
	\|\left\langle \cdot \right\rangle^k({V}^{\mathcal{E}}(z)-I)\|_{L^p(\Sigma^{ \mathcal{E} })}=\mathcal{O}(t^{-1/2}),\quad p\in[1,\infty],\,\,k\ge0.
\end{equation}

This uniformly vanishing bounded on ${V}^{\mathcal{E}}(z)-I$ establishes RH problem \ref{RH5.7} as a small-norm RH problem,
 for which there is a well known existence and uniqueness theorem. In fact, we may write
\begin{equation}\label{eq:5.47}
	\mathcal{\mathcal{E}}(z) =I+\frac{1}{2\pi i}\int_{\Sigma^{ \mathcal{E} }}\frac{(I+\eta(s))({V}^{ \mathcal{E} }(s)-I)}{s-z}ds,
\end{equation}
where $\eta\in L^2(\Sigma^{ \mathcal{E} })$ is the unique solution of the equation
\begin{equation}\label{eq:334}
	(I-\mathcal{C} )\eta=\mathcal{C} I,
\end{equation}
where the integral operator $\mathcal{C}: L^2 \to L^2 $ is given  by
\begin{align}\label{eq:5.49}
	&\mathcal{C} f=\mathcal{P}^-(f( {{V}^{ \mathcal{E} }}-I)),
\end{align}
where $\mathcal{P}^-$ is the   Cauchy projection operator defined by
 \begin{align}\label{eq:5.49}
	& \mathcal{P}^-f(z)=\underset{\epsilon \to 0 }{\lim}\frac{1}{2\pi i}\int_{\Sigma^{ \mathcal{E}  }}\frac{f(s)}{s-(z-i\epsilon)}ds.
\end{align}

It then follows from equation \eqref{eq?OE?332} and \eqref{eq:5.49} that
\begin{equation}
	\|\mathcal{C}\|_{L^2(\Sigma^{ \mathcal{E} })}\lesssim\|\mathcal{P}^-\|_{L^2(\Sigma^{\mathcal{E}})}\|{V}^{\mathcal{E}}-I\|_{L^\infty(  \Sigma^{\mathcal{E}}  )} \lesssim\mathcal{O}(t^{-1/2}),
\end{equation}
which guarantees the existence of the resolvent operator $(I-\mathcal{C})^{-1}$ and thus of both $\eta$ and $\mathcal{E}(z)$.

We expand   $\mathcal{E}(z)$    at  $z= 0$
\begin{align}
	\mathcal{E}(z)&
	 =I+\mathcal{E}_0+\mathcal{E}_1z+\mathcal{O}(z^2),\nonumber
\end{align}
where
\begin{align}
		&\mathcal{E}_k=\frac{1}{2\pi i}\oint_{ \partial \mathcal{U}}\frac{  {V}^{\mathcal{E}}(s)-I }{s^k}+O(t^{-1})
 = \frac{A_k}{\sqrt{t}} +O(t^{-1}),\,\,k=0,1.\label{eq:346}
\end{align}
where
\begin{align}A_k:=\sum_{n=0}^2  \sum_{j=1}^2  \frac{
M^{out}( \varkappa_{nj})M_1^{pc}( \varkappa_{nj}) M^{out}( \varkappa_{nj})^{-1}}{\sqrt{c_{nj} } \varkappa_{nj}^k }, \ \label{eqAK}
\end{align}

\subsection{Analysis on the remaining $\bar\partial$-problem}\label{sec3.3}
\begin{problem}\label{RH5.1}
	Find a continuous with sectionally continuous first partial derivatives function ${M}^{(3)}(z):\mathbb{C}\to SL_3(\mathbb{C})$ with the following properties:
	
	$\blacktriangleright$\emph{{The normalization condition: }}
	$${M}^{(3)}( z)=I+\mathcal{O}(z^{-1}),\quad z\to\infty.$$
	
	$\blacktriangleright$\emph{{The symmetry proposition}}	
	$${M}^{(3)}( z)=\Gamma_1\overline{{M}^{(3)}( \bar z)}\Gamma_1=\Gamma_2\overline{{M}^{(3)}( \omega^2\bar z)}\Gamma_2=\Gamma_3\overline{{M}^{(3)}( \omega\bar z)}\Gamma_3.$$	
	
	$\blacktriangleright$\emph{{The $\bar\partial$ condition:}}
	For $z\in\mathbb{C}$, we have
	\begin{equation}
		\bar\partial{M}^{(3)}( z)={M}^{(3)}( z){W}^{(3)}(z),
	\end{equation}
	where ${W}^{(3)}(z):=M^{rhp}(z) (z)\bar\partial J^{(2)} (z)M^{rhp}(z)^{-1}$.
\end{problem}
$\bar\partial$-RH problem \ref{RH5.1} is equivalent to the integral equation
\begin{equation}\label{eq:6.1}
	{M}^{(3)}(z) =I-\frac{1}{\pi}\iint_{\mathbb{C}}\frac{{M}^{(3)}(s){W}^{(3)}(s)}{s-z}dA(s),
\end{equation}
where $dA(s)$ is Lebesgue measure on the plane, equation (\ref{eq:6.1}) can be written using operator notation as
\begin{equation}
	(I-S){M}^{(3)}(z)=I,
\end{equation}
where $S$ is the solid Cauchy operator
\begin{equation}\label{eq:6.3}
	S(f)(z)=-\frac{1}{\pi}\iint_{\mathbb{C}}\frac{f(s){W}^{(3)}(s)}{s-z}dA(s).
\end{equation}

The subsequent Proposition \ref{pro6.1} demonstrates that for $t\to+\infty$, the operator $S$ is small-norm, so that the resolvent operator $(I-S)^{-1}$ exists and can be expressed as Neumann series.
\begin{proposition}\label{pro6.1}
	There exist a constant $c$ such that   the operator (\ref{eq:6.3}) satisfies the estimate
	\begin{equation}
		\|S\|_{L^\infty\to L^\infty}\leq ct^{-1/4},\quad t\to+\infty. 	\label{oor}
		\end{equation}
\end{proposition}
\begin{proof}
	We provide a detailed proof for (\ref{oor})  in the region $\Omega^0_{5+}$, and the same reasoning applies to other regions as well.
	
	Let $f\in L^\infty(\Omega^0_{5+})$, $s=u+iv$ and $z=\alpha+i\beta$.
	By using the boundness of  $\left\|M^{rhp} (z)  \right\|_{L^\infty }$ and
	the estimate  \eqref{eq:3133}, we can deduce that
	\begin{align}
		&|S(f)(z)|\leq\iint_{\Omega^0_{5+}}\frac{|f(s)M^{rhp} (s)\bar\partial J^{(2)} (s) M^{rhp} (s)^{-1}}{|s-z|}dA(s)\nonumber\\
		&\lesssim
		\|M^{rhp} (z)\|_{L^\infty }\|M^{rhp} (z)^{-1}\|_{L^\infty }\|f (z) \|_{L^\infty }
		\iint_{\Omega^0_{5+}}\frac{|\bar\partial  J^{(2)} (s)}{|s-z|}dA(s)\nonumber\\
		&\leq c(I_1+I_2+I_3)\|f (z) \|_{L^\infty }, \label{eq:6.5}
	\end{align}
	where  	
	\begin{align}\label{eq:6.7}
		&I_1=\iint_{\Omega^0_{5+}}\frac{|\bar\partial \chi_{\mathcal{Z}} +\varphi({\rm Re}\,z)|e^{-\sqrt{3}tv(u-\varkappa)}}{|s-z|}dA(s),\\\label{eq:6.8} &I_2=\iint_{\Omega^0_{5+}}\frac{|\bar r'(z)|e^{-\sqrt{3}tv(u-\varkappa)}}{|s-z|}dA(s),\\\label{eq:6.9}
		&I_3=\iint_{\Omega^0_{5+}}\frac{|z-\varkappa|^{-1/2}e^{-\sqrt{3}tv(u-\varkappa)}}{|s-z|}dA(s).
	\end{align}
Throughout we use the elementary fact
	\begin{align}\nonumber
		\Big\|\frac{1}{s-z}\Big\|_{L^2(v+\varkappa)}^2&=\Big(\int_{v+\xi}^\infty\frac{1}{(u-\alpha)^2+(v-\beta)^2}du\Big)^{1/2}\leq \frac{\pi}{|v-\beta|}
	\end{align} to shown that
	\begin{align}
		|I_1|\lesssim \int_0^\infty\frac{e^{-\sqrt{3}tv^2}}{|v-\beta|^{1/2}}dv\lesssim t^{-1/4}\int_{\mathbb{R}}\frac{e^{-\sqrt{3}(w+t\beta)^2}}{|w|^{1/2}}dw\lesssim t^{-1/4}. \label{p1}
	\end{align}
	The bound for $I_2$ is similar to $I_1$ that is
	\begin{align}
		|I_2|\lesssim\|\bar r'(z)\|_{L^2(\mathbb{R})}\int_0^\infty e^{-\sqrt{3}tv^2}\Big\|\frac{1}{s-z}\Big\|_{L^2(-\infty,v-\varkappa)}dv\lesssim t^{-1/4}. \label{p2}
	\end{align}
	For $I_3$ choose $p>2$ and $q$ H\"older conjugate to $p$, then
	\begin{align}
		|I_3|\lesssim\int_0^\infty e^{-\sqrt{3}tv^2}v^{1/p-1/2}|v-\beta|^{1/q-1}dv\lesssim t^{-1/4}. \label{p3}
	\end{align}
	The result is confirmed.
\end{proof}

To study the long time asymptotic behavior of $u(x,t)$,  it becomes essential to study the asymptotic expansion
of ${M}^{(3)}(z) $ as $z\to 0$,
\begin{align}\nonumber
	{M}^{(3)}(z) &=  I+{M}^{(3)}_0+{M}_1^{(3)}z+\mathcal{O}(z^2),
\end{align}
where
\begin{align}\label{eq:366}
	& {M}_0^{(3)}=-\frac{1}{\pi}\iint_{\mathbb{C}}\frac{ {M}^{(3)}(s) {W}^{(3)}(s)}{s}dA(s),\\\label{eq:369}
	&{M}^{(3)}_1=-\frac{1}{\pi}\iint_{\mathbb{C}}\frac{ {M}^{(3)}(s) {W}^{(3)}(s)}{s^2}dA(s).
\end{align}

\begin{proposition}\label{pro6.2}
There exists a constant $c$ such that
	\begin{equation}\label{eq:3522}
		|{M}_0^{(3)}|<c t^{-\frac{3}{4}}, \quad 	|{M}_1^{(3)}|<c t^{-\frac{3}{4}}.
	\end{equation}
\end{proposition}
\begin{proof} We take the case as $z\in\Omega^0_{5+}$ and $ j=0$  as illustrative  example. Consider the estimate \eqref{eq:4.2} and \eqref{eq:316}, we divide the integral region $\Omega^0_{5+}$ into $\Omega^c_{5+}=\Omega^0_{5+}\cap B(0)$ and  $\Omega^s_{5+}=\Omega^0_{5+}\setminus B(0)$
	where
	$B(0)=\{z\,|\,|z|\leq\rho/6\}.$ Let $s=u+iv$, $z=\alpha+i\beta,$ we acquire
	\begin{align*}
		| {M}_1^{(3)}|\lesssim\Big\{\iint_{\Omega^c_{5+}}	+\iint_{\Omega^s_{5+}}\Big\}\Big|\frac{\bar\partial R^0_2(s)e^{-2it\theta}}{s}\Big|dA(s)\leq c(I_4+I_5),
	\end{align*}
	where
	\begin{align}
		&I_4=\iint_{\Omega^c_{5+}}	\frac{|\bar\partial R^0_2(s)|e^{-\sqrt{3}tv(u-\varkappa)}}{|s|}dA(s),\\
		&I_5=\iint_{\Omega^s_{5+}}\frac{\bar\partial| R^0_2(s)|e^{-\sqrt{3}tv(u-\varkappa)}}{|s|}dA(s).
	\end{align}
	Then, using   \eqref{eq:316} we acquire
	\begin{align*}
		|I_4|&\leq\iint_{\Omega^c_{5+}}	\frac{|\bar\partial R^0_2(s)|e^{-\sqrt{3}tv(u-\varkappa)}}{|s|}dudv\lesssim \iint_{\Omega^c_{5+}}	e^{-\sqrt{3}tv^2}dudv
		\lesssim t^{-3/4}.
	\end{align*}
	To obtain the bound for $I_5$, we can express it as follows:
	\begin{align}
		I_5&=\iint_{\Omega^s_{5+}}\frac{|\bar\partial R^0_2(s)|e^{-\sqrt{3}tv(u-\varkappa)}}{|s|}dA(s)\leq  I^1_5+ I^2_5,
	\end{align}
	where
	\begin{align*}
		& I^1_5= \iint_{\Omega^s_{5+}}\frac{(|\bar r'(z)|+\bar\partial(\chi_{\mathcal{Z}}))e^{-\sqrt{3}tv(u-\varkappa)}}{|s|}dA(s),\\
		& I^2_5=\iint_{\Omega^s_{5+}}\frac{|z-\varkappa|^{-1/2}e^{-\sqrt{3}tv(u-\varkappa)}}{|s|}dA(s).
	\end{align*}
	For $ I^1_5$, we can utilize the estimate \eqref{eq:316} and observe that the function $|s|^{-1}$ is bounded, we have
	\begin{align*}
		| I^1_5|&\lesssim\iint_{\Omega^s_{5+}}|\bar r'(z)+\bar\partial(\chi_{\mathcal{Z}})|e^{-\sqrt{3}tv(u-\varkappa)}dudv\lesssim t^{-3/4}.
	\end{align*}
	For $I_5^2$ choose $2<p<4$, by   H\"older  inequality,
	\begin{align*}
		|I^2_5|\lesssim\int_0^\infty \||z-\varkappa|^{-1/2}\|_{L^p}\|e^{-\sqrt{3}tv^2}\|_{L^q}dv
		\lesssim t^{-3/4}. \end{align*}
	  The result is confirmed.
\end{proof}

\subsection{The proof of Theorem \ref{th1.1} for the  region I}\label{sec34}

In this subsection, we give the proof of Theorem \ref{th1.1} for the  region I.

From a series of transformations (\ref{eq:2.32}), (\ref{eq:312}), \eqref{eq:3.177} and (\ref{eq:320}),  we find  the solution of RH problem \ref{RH2.1} is given by
\begin{equation}
	M( z)= {M}^{(3)}( z)	\mathcal{E}(z) {M}^{out}( z) \mathcal{R}^{(2)}(z)^{-1} Y(z),
\end{equation}
Taking $z\to0$ along  $\pi/4$ so that $\mathcal{R}^{(2)}(z)=I$,
then we have
\begin{align*}
	M( z)&=(I+{M}^{(3)}_0+{M}_1^{(3)}z+\cdots)\cdot(I+\mathcal{E}_0+\mathcal{E}_1z+\cdots)\\
	&\times ({M}_0^{out}+{M}_1^{out} z+\cdots)\cdot(Y_0+Y_1z+\cdots).
\end{align*}

Further by using reconstruction formula (\ref{eq:2.39})-(\ref{eq:2.40}), we obtain
the soliton resolution for the OV equation as follows
\begin{align}
& u(x,t) = {u}_{sol}(y,t)+t^{-1/2}f_t(y,t)+\mathcal{O}(t^{-\frac{3}{4}}),\nonumber\\
&x =y+g(y,t)+t^{-1/2}f(y,t)+\mathcal{O}(t^{-3/4}),\nonumber
			\end{align}
where
\begin{align}\label{eqg}
	&
	g(y,t)=\frac{F_3^1}{F_3^0}+\frac{\sum_{j=1}^3[M_1^{out}]_{j3}}{\sum_{j=1}^3[M_0^{out}]_{j3}},\\\label{eqf}
&	f(y,t)= \frac{F_3^1}{F_3^0}+\frac{\sum_{j=1}^3[A_1M_0^{out}]_{j3}}{\sum_{j=1}^3[M_0^{out}]_{j3}}+\frac{\sum_{j=1}^3[A_0M_1^{out}]_{j3}}{\sum_{j=1}^3[M_0^{out}]_{j3}},
\end{align}
 with  $F_3^0$ and $F_3^1$  being  given by  \eqref{rr32}, whale $A_0$ and $A_1$ being given by  \eqref{eqAK}.

\section{Soliton Resolution  in the  Region  II: $x/t>0$ }\label{sec44}

In this section, we give the soliton resolution for the       OV  equation  (\ref{eq:1.1})  in  the asymptotic    region II.

\subsection{Opening $\bar\partial$-lenses}\label{sec41}

We open the contours at $z=0$ due to the absence of a phase  point in region $x/t>\epsilon$. Define
$$\Sigma^{(2)}=L\cup\omega L\cup\omega^2L,\quad L=\cup_{j=1}^4l^0_j.  $$
We define the regions between $l^0_j$ and ${\rm Re}\,z$ in a counterclockwise direction, as $\Omega^0_j$  for each $j=1,2,3,4$. Following this, we similarly define  $\Omega^1_j$ and $\Omega^2_j$
using the same method, see Figure\,\ref{F411}.
\begin{figure}[htp]
	\begin{center}
		\begin{tikzpicture}[scale=0.4]
			\draw[thick][YellowOrange!80](4.4,-0.6)--(7,-1);
			\draw[-latex,thick][YellowOrange!80](0,0)--(4.4,-0.6);
			\draw[thick][YellowGreen!90](-4.4,0.6)--(-7,1);
			\draw[-latex,thick][YellowGreen!90](0,0)--(-4.4,0.6);
			
			\draw[thick][YellowOrange!80](-4.4,-0.6)--(-7,-1);
			\draw[-latex,thick][YellowOrange!80](0,0)--(-4.4,-0.6);
			\draw[thick][YellowGreen!90](4.4,0.6)--(7,1);
			\draw[-latex,thick][YellowGreen!90](0,0)--(4.4,0.6);
			\draw[dashed ](-7,0)--(7,0);
			\draw[dashed](-5,5)--(5,-5);
			\draw[dashed](5,5)--(-5,-5);
			
			\draw[thick][YellowGreen!90](-2.1,-2.7)--(-3.8,-5);
			\draw[-latex,thick][YellowGreen!90](0,0)--(-2.1,-2.7);

			\draw[thick][YellowOrange!80](2.1,2.7)--(3.8,5);
			\draw[-latex,thick][YellowOrange!80](0,0)--(2.1,2.7);

			\draw[thick][YellowOrange!80](-2.1,2.7)--(-3.8,5);
			\draw[-latex,thick][YellowOrange!80](0,0)--(-2.1,2.7);

			\draw[thick][YellowGreen!90](2.1,-2.7)--(3.8,-5);
			\draw[-latex,thick][YellowGreen!90](0,0)--(2.1,-2.7);
			
			\draw[thick][YellowOrange!80](-3,-2.4)--(-6.3,-5);
			\draw[-latex,thick][YellowOrange!80](0,0)--(-3,-2.4);
			
			\draw[thick][YellowGreen!90](3,2.4)--(6.3,5);
			\draw[-latex,thick][YellowGreen!90](0,0)--(3,2.4);

			\draw[thick][YellowGreen!90](-3,2.4)--(-6.3,5);
			\draw[-latex,thick][YellowGreen!90](0,0)--(-3,2.4);
			
			\draw[thick][YellowOrange!80](3,-2.4)--(6.3,-5);
			\draw[-latex,thick][YellowOrange!80](0,0)--(3,-2.4);
			
			\draw(0,0)node[below]{\scalebox{0.8}{$0$}};
			\draw(7.8,1.7)node[below]{\scalebox{0.8}{$l^0_{1}$}};
			\draw(7.8,-0.5)node[below] {\scalebox{0.8}{$l^0_{4}$}};
			\draw(-7.8,1.7)node[below]{\scalebox{0.8}{$l^0_{2}$}};
			\draw(-7.8,-0.5)node[below]{\scalebox{0.8}{$l^0_{3}$}};
			\draw(7.1,6) node[below]{\scalebox{0.8}{$l^2_{4}$}};
			\draw(3.8,6.3)node[below]{\scalebox{0.8}{$l^2_{1}$}};
			\draw(-7.1,6) node[below]{\scalebox{0.8}{$l^1_{1}$}};
			\draw(-3.8,6.3)node[below] {\scalebox{0.8}{$l^1_{4}$}};
			\draw(-7.1,-4.6) node[below] {\scalebox{0.8}{$l^2_{2}$}};
			\draw(-3.8,-5)node[below]{\scalebox{0.8}{$l^2_{3}$}};
			\draw(7.1,-4.6) node[below]{\scalebox{0.8}{$l^1_{3}$}};
			\draw(3.8,-5)node[below]{\scalebox{0.8}{$l^1_{2}$}};
			
		\end{tikzpicture}
	\end{center}
	\caption{The opened contour $\Sigma^{(2)}$  and  regions $\Omega^n_j, n=0,1,2; j=1,2,3,4$.}
	\label{F411}
\end{figure}
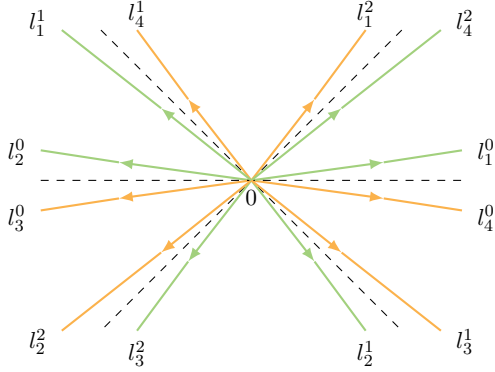

The next key step is to extend the jump matrix $J^{(1)}(z)$ continuously from the jump contours $\Sigma$ to the new contour $\Sigma^{(2)}$, along which the jump matrices decay.
The extension function $R^n_j(z),\ n=0,1,2;j=1,2,3,4$  are defined  as follows.

\begin{proposition} \label{pro12}
	There  exist functions  $  R^n_j(z):\bar \Omega_j^n \to\mathbb{C}, n=0,1,2; j=1,2,3,4$ satisfying
	\begin{align*}\nonumber
		&  R^n_1(z)=\begin{cases}
		p^n_1(z)G_n(z)^{-1},\, z\in \omega^n\mathbb{R}^+,\\
			0,\hspace{2.2cm} z\in l^n_{1},
		\end{cases}  R^n_2(z)=\begin{cases}
			p^n_2(z)G_n(z)^{-1},\,z\in\omega^n\mathbb{R}^-,\\
			0,\hspace{2.2cm} z\in l^0_{2},
		\end{cases}\\
		&	  R^n_3(z)=\begin{cases}
			p^n_3(z)G_n(z),\,\,\,\,\,\,\,\,\,z\in\omega^n\mathbb{R}^-\\
			0,\hspace{2.3cm}   z\in l^0_{3},
		\end{cases}	  R^n_4(z)=\begin{cases}
			p^n_4(z)G_n(z),\,\,\, \,\,\,\,\,\,\,z\in \omega^n\mathbb{R}^+,\\
			0,\hspace{2.3cm} z\in l^0_{4},
		\end{cases}
	\end{align*}
where $p_{n,j}(z)$ and $G_n(z)$ are defined in \eqref{eq:333} and \eqref{eq:3334}, respectively,
	such that  the  functions $  R^n_j(z)$ have the following properties:
	\begin{align}\label{eq:433}
		&	|\bar\partial   R^n_j(z)|\lesssim |\varphi({\rm Re}\,z)|+| r'({\rm Re}\,z)|+|z|^{-1/2}+|\bar\partial\chi_{\mathcal{Z}}(z)|,\quad j=1,2,3,4,
	\end{align}
	where  $\varphi\in C_0^\infty(\mathbb{R}, [0,1])$  is a fixed cutoff function near small support near $z=0$.	
	Moreover, for $z\to0$, we have
	\begin{equation}\label{eq:43}
		|\bar\partial   R^n_j(z)|\lesssim |z|^2.
	\end{equation}
	
\end{proposition}
\begin{proof}
This proof is similar to Lemma \ref{le4.1}.
\end{proof}

Now we  define  a matrix function

\begin{align}\nonumber
	\mathcal{R}^{(2)}(z)=\begin{cases}
		\begin{pmatrix}
			1&0&0\\
			R^0_je^{2it\theta(z)}&1&0\\
			0&0&1
		\end{pmatrix},\,\,z\in\Omega^0_{j=3,4},\,\,\,\,\,
		\begin{pmatrix}
			1&R^0_{j}e^{-2it\theta(z)}&0\\
			0&1&0\\
			0&0&1
		\end{pmatrix},\,\,\,\,z\in\Omega^0_{j=1,2},\\[15pt]
		\begin{pmatrix}
			1&0&0\\
			0&1&0\\
			R^1_je^{2it\theta(\omega^2z)}&1&0\\
		\end{pmatrix},\,z\in\Omega^1_{j=3,4},
		\begin{pmatrix}
			1&0&R^1_{j}e^{-2it\theta(\omega^2z)}\\
			0&1&0\\
			0&0&1
		\end{pmatrix},\,z\in\Omega^1_{j=1,2},\\[15pt]
		\begin{pmatrix}
			1&0&0\\
			0&1&	R^2_je^{2it\theta(\omega z)}\\
			0&0&1
		\end{pmatrix},\,\,\,\,z\in\Omega^2_{j=3,4},
		\begin{pmatrix}
			1&0&0\\
			0&1&0\\
			0&R^2_{j}e^{-2it\theta(\omega z)}&1
		\end{pmatrix},\,z\in\Omega^0_{j=1,2},\\[15pt]
		I\qquad\hspace{3cm} z\in\,elsewhere.
	\end{cases}
\end{align}
then new function
\begin{equation}\label{eq44.3}
	 {{M}}^{(2)} (z)=M(z)
	\mathcal{R}^{(2)}(z)
\end{equation}
satisfies the following  mixed $\bar\partial$-RH problem.
\begin{problem}\label{RH411}
	Find a analytic function $M^{(2)}(z):\mathbb{C}\setminus \mathcal{Z} \to SL_3(\mathbb{C})$ with the following properties:
	
	$\blacktriangleright$\emph{{The normalization condition:}}
	$$ {{M}}^{(2)}( z)=I+\mathcal{O}(z^{-1}),\quad z\to\infty.$$

	$\blacktriangleright$ ${{M}}^{(2)}(z)$ satisfies  the same  residue conditions with ${M}^{(1)}( z)$.

	$\blacktriangleright$ 	
	For each $z\in\mathbb{C}\setminus \mathcal{Z}$, we have
	\begin{align}\label{eq:444}
		\bar\partial {{M}}^{(2)}( z) = {{M}}^{(2)}(z)\bar\partial\mathcal{R}^{(2)}(z).
			\end{align}
\end{problem}

We perform the following factorization
\begin{equation}\label{eq:3.17}
	{M}^{(2)}( z) =	M^{(3)}( z)M^{out}(z),
\end{equation}
where    $M^{out}(z)$  is the solution  of  the pure soliton RH problem \ref{RH39}
and ${M}^{(3)}(z)$ is a continuously   function that satisfies the  $\bar\partial$-RH problem \ref{RH41}.

\subsection{Analysis on the remaining $\bar\partial$-problem}\label{sec4.3}

\begin{problem}\label{RH41}
	Find a continuous   function $	M^{(3)}(z): \mathbb{C}\to SL_3(\mathbb{C})$ with the following properties.
	
	$\blacktriangleright$\emph{{The normalization condition:}}
	$$	M^{(3)}( z)=I+\mathcal{O}(z^{-1})\quad z\to\infty.$$
	
	$\blacktriangleright$ 	 $M^{(3)}( z) $  satisfies the  $\bar\partial$-equation:
	\begin{equation}\label{eq:4266}
		\bar\partial 	M^{(3)}( z)=	M^{(3)}( z)	W^{(3)}( z), \  z\in\mathbb{C},
	\end{equation}
	where $	W^{(3)}( z):=M^{out} ( z)\bar\partial R^{(2)}(z)M^{out} ( z)^{-1}.$
\end{problem}

$\bar\partial$-RH problem \ref{RH41} is equivalent to the integral equation
\begin{equation}\label{eq:61}
	M^{(3)}(z) =I-\frac{1}{\pi}\iint_{\mathbb{C}}\frac{M^{(3)}(s){W}^{(3)}(s)}{s-z}dA(s),
\end{equation}
which can be  written   as  an operator  equation
\begin{equation}
	(I-S)M^{(3)}(s)=I,
\end{equation}
where $S$ is the solid Cauchy operator
\begin{equation}\label{eq:4416}
	S(f)(z)=-\frac{1}{\pi}\iint_{\mathbb{C}}\frac{f(s){W}^{(3)}(s)}{s-z}dA(s).
\end{equation}

We show   that for $t\to+\infty$, the operator $	{S}$ is small-norm, so that the resolvent operator
 $(I-	S)^{-1}$ exists.

\begin{proposition}\label{pro44}
	There exist a constant $c$ such that  the operator (\ref{eq:4416}) satisfies the estimate
	\begin{equation}
		\|S\|_{L^\infty\to L^\infty}\leq ct^{-1/2},\quad t\to+\infty.
	\end{equation}
\end{proposition}
\begin{proof}
	We show the case  in the region $\Omega^0_2$,     the other regions  follow  similarly. Let $f\in L^\infty(\Omega^0_2)$,  $s=u+iv$ and $z=\alpha+i\beta$, then it follows  \eqref{eq:433} and \eqref{eq:4266}, we get
	\begin{align}\label{eq:655}
		&|S(f)(s)|\leq\iint_{\Omega^0_2}\frac{|f(s)M^{out} (s)\bar\partial R^{(2)}(z)M^{out} (s)|^{-1}}{|s-z|}dA(s)\\\nonumber
		&\lesssim
		 (  I_1+  I_2+  I_3) \|f(z)\|_{L^\infty (\Omega^0_2)},
	\end{align}
where
	\begin{align}\label{eq:420}
		&  {I}_1=\iint_{\Omega^0_2}\frac{|\chi_{\mathcal{Z}}(s)+\varphi({\rm Re}\,s)|e^{-\sqrt{3}tv(\xi+\frac{1}{u^2+v^2})}}{|s-z|}dA(s),\\ &  {I}_2=\iint_{\Omega^0_2}\frac{|r'({\rm Re}\,s)|e^{-\sqrt{3}tv(\xi+\frac{1}{u^2+v^2})}}{|s-z|}dA(s),\\\label{eq:4232}
		&  {I}_3=\iint_{\Omega^0_2}\frac{|s-1|^{-1/2}e^{-\sqrt{3}tv(\xi+\frac{1}{u^2+v^2})}}{|s-z|}dA(s).
	\end{align}
Similar to (\ref{p1})- (\ref{p3}), direct calculation shows  that
	\begin{align*}
	&	|  {I}_1|\lesssim t^{-1/2}\int_{\mathbb{R}}\frac{e^{-\sqrt{3}(w+t\beta)}}{|w|^{1/2}}dw\lesssim t^{-1/2}\int_{\mathbb{R}}\frac{e^{-\sqrt{3}w}}{|w|^{1/2}}dw\lesssim t^{-1/2},\\
		&	|  {I}_2|\lesssim\int_{-\infty}^{-v+1} e^{-\sqrt{3}tv}\int^{-v+1}_{-\infty}\frac{|r'({\rm Re}\,s)|}{|s-z|}dudv\lesssim t^{-1/2},\\
		& 	|  I_3|\lesssim\int_0^\infty e^{-\sqrt{3}t\xi v}v^{1/p-1/2}|v-\beta|^{1/q-1}dv\lesssim t^{-1/2}.
	\end{align*}
	The result is confirmed.
\end{proof}
To recover the long time asymptotic behavior of $u(x,t)$, it is necessary to determine the asymptotic expansion of
 $M^{(3)}(z)$ as $z\to0$,
\begin{equation}
	M^{(3)}(z)=I+	M^{(3)}_0+M_1^{(3)}z+\mathcal{O}(z^2),
\end{equation}
where
\begin{align}
	&M^{(3)}_0=-\frac{1}{\pi}\iint_{\mathbb{C}}\frac{ M^{(3)}(s)W^{(3)}(s)}{s}dA(s),\\
	&M^{(3)}_1=-\frac{1}{\pi}\iint_{\mathbb{C}}\frac{ M^{(3)}(s) W^{(3)}(s)}{s^2}dA(s).
\end{align}

\begin{proposition}
	There exists a constant $c$ such that
	\begin{equation}
		| M^{(3)}_0|< ct^{-1},\quad 	| {{M}}^{(3)}_1|< ct^{-1}.
	\end{equation}
\end{proposition}
\begin{proof}

	Let $s=u+iv$ and $z=\alpha+i\beta$.    we divide the integral region $\Omega^0_{2}$ into $\Omega^c_{2}=\Omega^0_{2}\cap B(0)$ and $\Omega^s_{2}=\Omega^0_2\setminus B(0)$
	where
	$B(0)=\{z\,|\,|z|\leq\rho/6\}$, we have	
	\begin{align}\label{eq:4220}
		| {{M}}^{(3)}_0|&\leq\Big|\frac{1}{\pi}\iint_{\Omega^0_2}\frac{M^{(3)}(s) {{W}}^{(3)}(y,t;s)}{|s|}dA(s)\Big| < I_4+ I_5,
	\end{align}
	where
	\begin{align*}
		 & I_4=\iint_{\Omega^c_2}\frac{|\bar\partial R_2(s)|e^{-\sqrt{3}tv(\xi+\frac{1}{u^2+v^2})}}{|s|}dA(s),\\
   &  I_5= \iint_{\Omega^s_2}\frac{|\bar\partial R_2(s)|e^{-\sqrt{3}tv(\xi+\frac{1}{u^2+v^2})}}{|s|}dA(s). \end{align*}	
	Taking into account the estimate \eqref{eq:43}, we derive
	\begin{equation*}
		  I_4=\iint_{\Omega^c_2}\frac{|\bar\partial R^0_+(s)|e^{-\sqrt{3}tv(\xi+\frac{1}{u^2+v^2})}}{|s|}dA(s)\lesssim \iint_{\Omega^c_2}e^{-\sqrt{3}tv}\lesssim t^{-1}.    \end{equation*}
	Consider the estimate  \eqref{eq:433}, we obtain
	\begin{align*}
		  I_5&= \iint_{ \Omega^s_2}\frac{|\bar\partial R^0_+(s)|e^{-\sqrt{3}tv(\xi+\frac{1}{u^2+v^2})}}{|s|}dA(s) \leq   I^1_5+  I^2_5+  I^3_5,
	\end{align*}where
	\begin{align*}
		&  I^1_5=\iint_{ \Omega^s_2}\frac{|\chi_{\mathcal{Z}}(s)|e^{-\sqrt{3}tv(\xi+\frac{1}{u^2+v^2})}}{|s|}dA(s),\\ &  I^2_5=\iint_{ \Omega^s_2}\frac{|r'({\rm Re}\,s)|e^{-\sqrt{3}tv(\xi+\frac{1}{u^2+v^2})}}{|s|}dA(s),\\\label{eq:422}
		&  I^3_5=\iint_{ \Omega^s_2}\frac{|s-1|^{-1/2}e^{-\sqrt{3}tv(\xi+\frac{1}{u^2+v^2})}}{|s|}dA(s).
	\end{align*}
	Taking into consideration the bounded $|s|^{-1}$ within the domain $\Omega^s_2$, we obtain
	\begin{align*}
	&	| I^1_5|\lesssim\int_0^\infty e^{-\sqrt{3}tv}\int_{-\infty}^{-v+1}|\chi_\mathcal{Z}(s)|dudv
	\lesssim t^{-1},\\
&		|  I^2_5|\lesssim\int_0^\infty e^{-\sqrt{3}tv}\int_{-\infty}^{-v+1}\|r'({\rm Re}\,s)\|_{L^2(-\infty,-v+1)}dudv\lesssim t^{-1},\\
		& 	|  I^3_5|\lesssim\int_0^\infty v^{1/p-1/2}\left(\int_{-\infty}^{-v+1}e^{-\sqrt{3}tqv}du\right)^{1/q}dv
		\lesssim t^{-1}.
	\end{align*}
	The result is confirmed.
\end{proof}

\subsection{The proof of Theorem \ref{th1.1} for the  region II}\label{sec4.4}
In the region II, inverting the sequence of transformations (\ref{eq:2.32}), \eqref{eq44.3} and (\ref{eq:3.17})the solution of RH problem \ref{RH2.1} is given by
\begin{equation}
	M(z)=M^{(3)}(z)	  {M}^{out}(z)R^{(2)}(z)^{-1}Y(z).
\end{equation}
Taking $z\to0$ from $\pi/4$ so that $R^{(2)}(z)=I$,   we have
\begin{align}\nonumber
	M(z)&=(I+M^{(3)}_0+M_1^{(3)}z+\cdots)\cdot(I+ {\mathcal{E}}_0+ {\mathcal{E}}_1z+\cdots)\\
	&\times({M}^{out}_0+{M}^{out}_1z+\cdots)\cdot(Y_0+Y_1z+\cdots),
\end{align}

By using reconstruction formula (\ref{eq:2.39})-\eqref{eq:2.40},   we obtain the soliton resolution for
the OV equation
\begin{align*}
&				u(x,t) =u_{sol}(y,t;  {\mathcal{D}}   )+\mathcal{O}(t^{-1}),\\
&x =y+g(y,t)+\mathcal{O}(t^{-1}), \label{porre}
			\end{align*}
where  expressions for $g(y,t)$is provided in   \eqref{eqg}.

\begin{appendices}
	\section{The parabolic cylinder model problem}\label{secA}
	\renewcommand\thesection{A}
	\renewcommand{\thefigure}{A.\arabic{figure}}
	\renewcommand{\theequation}{A.\arabic{equation}}
	Here we describe the solution of the parabolic cylinder  RH model  problem which is appears frequently in the literature of long-time asymptotic calculations for integrable nonlinear waves
	\cite{MR18,RP18,DP93}.
	
	For   $r_0\in \mathbb{C}$ and $|r_0|<1$, let $\nu=-\frac{1}{2\pi} \log (1-|r_0|^2)$. Define  $ L^{pc}= \sum_{j=1}^4 L^{pc}_j $, see Figure\,\ref{FA.1},
	where  $L^{pc}_j $ denotes the complex contour
	$$L^{pc}_j =\left\{\zeta\in \mathbb{C}: \arg \zeta =\frac{(2j-1)\pi}{4}\right\}, \ j=1,2,3,4. $$

	\begin{figure}[htp]
		\begin{center}
			\begin{tikzpicture}[scale=0.5]
				\draw[dashed,-latex ](-5,0)--(5,0);		
				\node at (6,0){ Re\,$\zeta$};
				\draw[YellowGreen!90,thick][-latex](-4,3)--(-1.5,1.1);
				\draw[YellowGreen!90,thick](-1.5,1.1)--(1.5,-1.1);
				\draw[YellowGreen!90,thick][-latex](4,-3)--(1.5,-1.1);
				\draw[YellowOrange!90,thick][-latex](-4,-3)--(-1.5,-1.1);
				\draw[YellowOrange!90,thick](-1.5,-1.1)--(1.5,1.1);
				\draw[YellowOrange!90,thick][-latex](4,3)--(1.5,1.1);
				\node  [below]  at (0,-0.4) {$0$};
				\draw(-4.4,3)node[below] {$L_1^{pc}$};
				\draw(-4,-3)node[below] {$L_3^{pc}$};
				\draw(4.3,3)node[below] {$L_2^{pc}$};
				\draw(4.3,-3)node[below] {$L_4^{pc}$};
			\end{tikzpicture}
		\end{center}
		\caption{The contour $   \Sigma^{pc}=\cup_{j=1}^4L_j^{pc}$.}
		\label{FA.1}
	\end{figure}
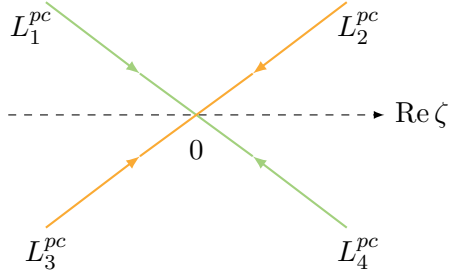
	\begin{problem} \label{RH5.3}
		Find a 3$\times$3 analytic matrix-valued function ${M}^{pc}( \zeta ):={M}^{pc}(r_0, \zeta )$\,:\, $\mathbb{C}\setminus  \Sigma^{pc}\to SL_3(\mathbb{C})$  with the following properties:
		
		$\blacktriangleright$ \emph{{ the normalization condition: }}
		$${M}^{pc}( \zeta)=I+\mathcal{O}(\zeta^{-1})\quad \zeta\to\infty.$$
		
		$\blacktriangleright$ \emph{{ the jump condition: }} ${M}^{pc}( \zeta)$ is continuous boundary value ${M}^{pc}_\pm( \zeta)$ on
		$  \Sigma^{pc}$, which satisfy the  jump relation
		$${M}^{pc}_+( \zeta) = {M}^{pc}_-( \zeta) {V}^{pc} ( \zeta),$$
		where
		\begin{align}
			{V}^{pc} (r_0,\zeta)=
			\begin{cases}
				\begin{aligned}
					& \begin{pmatrix}
						1&r_0\zeta^{2i\nu}e^{-i\zeta^2/2}&0\\
						0&1&0\\
						0&0&1
					\end{pmatrix},\zeta\in L_1^{pc},\quad
					\begin{pmatrix}
						1&0&0\\
						\frac{ \bar r_0\zeta^{-2i\nu}e^{i\zeta^2/2}}{ 1-|  r_0|^2 }&1&0\\
						0&0&1
					\end{pmatrix}, \zeta\in L_2^{pc},\\[2pt]
					&	\begin{pmatrix}
						1&0&0\\
						-\bar r_0\zeta^{-2i\nu}e^{i\zeta^2/2}&1&0\\
						0&0&1
					\end{pmatrix},\zeta\in L_3^{pc},\,\,
					\begin{pmatrix}
						1&-\frac{r_0\zeta^{2i\nu}e^{-i\zeta^2/2}}{ 1-| r_0|^2 }&0\\
						0&1&0\\
						0&0&1
					\end{pmatrix},\zeta\in L_4^{pc}.
				\end{aligned}
			\end{cases}\nonumber
		\end{align}
	\end{problem}
	
	The   RH problem \ref{RH5.3}  admits a unique   solution    the following asymptotics
	
	\begin{equation}
		{M}^{pc}(  \zeta)=I+\frac{ {{M}}^{pc}_1(r_0)   }{\zeta}+\mathcal{O}(\zeta^{-2}),
	\end{equation}
	where
	\begin{align}
{{M}}^{pc}_1  =\begin{pmatrix}
		0&-i\beta_{12}&0\\
		i\beta_{21}&0&0\\
		0&0&0
	\end{pmatrix}, \label{pc1}
\end{align}
	with
	\begin{equation}\label{eq:a5}
		\beta_{12}=\frac{\sqrt{2\pi}e^{i\pi/4}e^{-\pi\nu/2}}
		{r_0\Gamma(-i\nu)},\,\,\,\beta_{21}=\frac{-\sqrt{2\pi}e^{-i\pi/4}e^{-\pi\nu/2}}
		{\bar r_0\Gamma(i\nu)},
	\end{equation}
	where $\Gamma$ is the Gamma function.
\end{appendices}
\begin{appendices}
	\renewcommand\thesection{B}
	\section{Solutions  for reflectionless scattering data}\label{secB}
	
	\renewcommand{\thefigure}{B.\arabic{figure}}
	\renewcommand{\theequation}{B.\arabic{equation}}
	Here we are considering the solutions to the RH problem \ref{RH2.1} associated with the OV equation \eqref{eq:1.1} with no reflections
	\begin{align}
		\hat{\mathcal{D}}=\{ r(z)\equiv0,(\xi_n, \hat c_n)_{n=1}^{6N}\},\label{dsf}
	\end{align}
	with
	\begin{equation}
		\hat c_n= c_n\exp\left( \frac{i}{\pi}\int_{{I}}\frac{\log(1-|r(s)|^2)}{s-\xi_n}ds\right).
	\end{equation}
	The resulting solution known as  reflectionless  solutions, denoted by $u_{sol}(x,t;\hat{\mathcal{D}})$, are $N$-soliton solutions. When $N=1$, the single loop soliton solution of the OV  equation  \eqref{eq:1.1} is given by \cite{AD15}
	
	\begin{align}\label{eq:1.10}
		\mathcal{Q}_{sol} (x,t)=	\hat{\mathcal{Q}}_{sol} (y(x,t),t)=\frac{12}{\rho^2}\frac{\hat{e}(\cos(\phi+\frac{\pi}{3})-\hat{e}+\cos(\phi+\frac{\pi}{3})\hat{e}^2)}{(1-4\cos(\phi+\frac{\pi}{3})\hat{e}+\hat{e}^2)^2},
	\end{align}
	where
	\begin{align*}
		&\hat{e}(y,t)=\frac{c}{2\sqrt{3}\rho}e^{-\sqrt{3}\rho(y+\frac{t}{\rho^2}-\frac{1}{\sqrt{3}\rho}\log\frac{\hat c}{2\sqrt{3}\rho})},	\\
		&x(y,t) =y+\frac{2\sqrt{3}}{\rho}\frac{-2\cos(\phi+\frac{\pi}{3})\hat e+\hat e^2}{1-4\cos(\phi+\frac{\pi}{3})\hat e+\hat e^2}.
	\end{align*}
When $N>1$, the solution formulas become ungainly. However, we usually expect that for $t\to+\infty$, the solution will consist of $N$ independent $1$-soliton. For this reason, we now denote the solution of the reflectionless RH problem \ref{RHB.1} associated with equation  \eqref{eq:1.1}.

	\begin{problem}\label{RHB.1}
		Given   scattering data (\ref{dsf}),
		find an analytic function ${M}^{sol}(z;\hat{\mathcal{D}}):\mathbb{C}\setminus\mathcal{Z} \to SL_3(\mathbb{C})$ such that
		
		$\blacktriangleright$ \emph{{The normalization condition: }}
		$${M}^{sol}(z;\hat{\mathcal{D}})=I+\mathcal{O}(z^{-1})\quad z\to\infty.$$
		
		$\blacktriangleright$\emph{{ The residue condition:}}
		${M}^{sol}(z;\hat{\mathcal{D}})$ satisfies  the same  residue conditions with ${M}^{(1)}( z)$.
	
	\end{problem}

	Further we show that
	
	\begin{proposition}\label{prob1}
		Given   scattering data (\ref{dsf}) such that $z_j\ne z_k$ for $j\ne k$, there is a unique solution to the RH problem \ref{RHB.1} for each $(x,t)\in\mathbb{R}\times\mathbb{R}^+$.
	\end{proposition}
	\begin{proof}
		
		We will now only consider the case  for $\xi_n\in\Omega_1$, $n=1,\dots,N_1$.
		It follows that  the  solution of the RH problem \ref{RHB.1}   admits   the  expansion
		\begin{equation}\label{eq:B.5}
			{M}^{sol}(z;\hat{\mathcal{D}})=I+\sum_{n=1}^{N_1}\Big(\frac{\mathcal{Q}_1}{z-\xi_n}+\frac{\mathcal{Q}_2}
			{z-\omega\bar\xi_n}+\frac{\mathcal{Q}_3}{z-\omega\xi_n}+\frac{\Gamma_1\bar{\mathcal{Q}}_1\Gamma_1}
			{z-\bar\xi_n}+\frac{\Gamma_1\bar{\mathcal{Q}}_2\Gamma_1}{z-\omega^2\xi_n}+\frac{\Gamma_1\bar{\mathcal{Q}}_3\Gamma_1}{z-\omega^2\bar\xi_n}\Big),
		\end{equation}
		where
		\begin{align*}
			&\mathcal{Q}_1=\begin{pmatrix}
				0&\alpha^1_n&0 \\
				0&\beta^1_n&0 \\
				0&\gamma^1_n&0 \\
			\end{pmatrix},\mathcal{Q}_2=\begin{pmatrix}
				0&0&\alpha^2_n \\
				0&0&\beta^2_n \\
				0&0&\gamma^2_n \\
			\end{pmatrix} ,\mathcal{Q}_3=\begin{pmatrix}
				\alpha^3_n&0&0 \\
				\beta^3_n &0&0\\
			\gamma^3_n &0&0\\
			\end{pmatrix} ,
		\end{align*}
		in which  the coefficients $\alpha^j_n $, $\beta^j_n(x,t)$  and $\gamma^j_n(x,t)$,   $j=1,2,3$; $n=1,\dots,N_1$ are  to be determined.
		
		Inserting the partial fraction expansion \eqref{eq:B.5} into the residue conditions in RH problem \ref{RHB.1} leads to  the following linear system of equations
		\begin{align}
			&	\hat\alpha^1_j-\sum_{n=1}^{N_1}\Big(\frac{\sqrt{\varrho^1_j\bar{\varrho}^2_n}}{\xi_j-\omega^2\xi_n}\hat{\bar\beta}^2_n
			-\frac{\sqrt{\varrho^1_j\bar\varrho^3_n}}{\xi_j-\omega^2\bar\xi_n}\hat{\bar\beta}^3_n\Big)=\sqrt{\varrho^1_j},\label{eq:B5}\\
			&	\hat\alpha^2_j-\sum_{n=1}^{N_1}\Big(\frac{\sqrt{\varrho^2_j\varrho^1_n}}{\xi_j-\xi_n}\hat{\alpha}^1_n
			-\frac{\sqrt{\varrho^2_j\bar\varrho^1_n}}{\xi_j-\bar\xi_n}\hat{\bar\beta}^1_n\Big)=0,\\
			&
			\hat{\alpha}^3_j-\sum_{n=1}^{N_1}\Big(\frac{\sqrt{\varrho^3_j\bar\varrho^2_n}}{\xi_j-\omega^2\xi_n}\hat{\bar\beta}^2_n-
			\frac{\sqrt{\varrho^3_j\bar\varrho_n^3}}{\xi_j-\omega^2\bar\xi_n}\hat{\bar\beta}^3_n\Big)=\sqrt{\varrho_j^3},\\	 	
			&	\hat{\bar\beta}^1_j-\sum_{n=1}^{N_1}\Big(\frac{\sqrt{\bar\varrho^1_j\varrho^2_n}}{\bar\xi_j-\omega\bar \xi_n}\hat{\alpha}^2_n -\frac{\sqrt{\bar\varrho^1_j\bar\varrho^3_n}}{\bar\xi_j-\omega\xi_n}\hat{\alpha}^3_n\Big)=0,\\
			&\hat{\bar\beta}^2_j-\sum_{n=1}^{N_1}\Big(\frac{\sqrt{\bar\varrho^2_j\bar\varrho^1_n}}{\xi_j-\xi_n}\hat{\bar\beta}^1_n
			-\frac{\sqrt{\bar\varrho^2_j\varrho^1_n}}{\bar\xi_j-\xi_n}\hat{\alpha}^1_n\Big)=0,\\
			&\hat{\bar\beta}^3_j-\sum_{n=1}^{N_1}\Big(\frac{\sqrt{\bar\varrho^3_j\varrho^2_n}}{ \xi_j-\omega\bar\xi_n}\hat{\alpha}^2_n-\frac{\sqrt{\bar\varrho^3_j\varrho_n^3}}{ \xi_j-\omega\xi_n}\hat{\alpha}^3_n\Big)=0, \ j=1,\dots,N_1,\label{eq:B10}
		\end{align}		
		where we have used the renormalized parameters
		\begin{align*}
			\hat\alpha^i_n=\alpha^i_n/\sqrt{\varrho_n^i},\quad\hat{\bar\beta}^i_n=\bar\beta^i_n/\sqrt{\bar\varrho_n^i},\quad\hat{\gamma}^i_n=\gamma^i_n/\sqrt{\varrho_n^i},\quad\hat{\bar\gamma}^i_n
			=\bar\gamma^i_n/\sqrt{\bar\varrho_n^i}.
		\end{align*}
		Further let
		$$ \vec{\alpha}=(\hat\alpha_1,\hat\alpha_2,\hat\alpha_3)^T, \ \
		\vec{\beta}  =(\hat{\bar\beta}_1,\hat{\bar\beta}_2,\hat{\bar\beta}_3)^T, \ \  \vec{\varrho} =(\sqrt{\varrho_1},0,\sqrt{\varrho_3})^T,
		$$
		where
		\begin{align*}
			&		
			\hat\alpha_j=\begin{pmatrix}
				\hat\alpha^j_1,\dots,\hat\alpha_{N_1}^j
			\end{pmatrix}^T	, \ \ \hat{\bar\beta}_j=\begin{pmatrix}
				\hat{\bar\beta}^j_1,\dots,\hat{\bar\beta}^j_{N_1}
			\end{pmatrix}^T	, \ \ \sqrt{\varrho}_j=\begin{pmatrix}
				\sqrt{\varrho^j_1},\dots,\sqrt{\varrho^j_{N_1}}
			\end{pmatrix}^T,
		\end{align*}	
		with $ j=1,2,3$, then system \eqref{eq:B5}-\eqref{eq:B10} can be written in  the block matrix form
		\begin{equation}\label{eq:B14}
			\begin{pmatrix}
				I&iA\\
				-iA^H&C\\
			\end{pmatrix}\begin{pmatrix}
				\vec{\alpha}\\
				\vec{ \beta}
			\end{pmatrix}=\begin{pmatrix}
				\vec{\varrho} \\
				0
			\end{pmatrix},
		\end{equation}
		where $A, B$ and $ C$  are  $3N_1\times3N_1$ matrices
		\begin{align}
			A=\begin{pmatrix}
				0&*&*\\
				*&0&0\\
				0&*&*
			\end{pmatrix},\  B=\begin{pmatrix}
				*&*&0\\
				*&0&0\\
				0&*&*
			\end{pmatrix}, \ C= \begin{pmatrix}
				I_{N_1}&0&0\\
				*&I_{N_1}&0\\
				0&0&I_{N_1}
			\end{pmatrix},
		\end{align}
		with the signal  $*$ stands for a $N_1\times N_1$ Hermitian  matrix.
		It can be shown that the matrix  $A\bar A $ is  positive definite
		and its eigenvalues   $\{\lambda_k\}_{j=1}^{N_1}\subset\mathbb{R}_+,$
		which implies  that
		\begin{equation*}
			\det(I_{N_1}+A\bar A)=\prod_{k=1}^{N_1}(1+\lambda_k)>0,
		\end{equation*}
		Hence    the system  \eqref{eq:B14} has  a unique solution.
		
		On the other hand,   $\gamma^j_n $ and $\bar\gamma^j_n $  can be uniquely obtained from  a  similar  system with \eqref{eq:B5}-\eqref{eq:B10}.
	\end{proof}
\end{appendices}
\vspace{3mm}	

\noindent{Acknowledgements}

This work is supported by  the National Natural Science
Foundation of China (Grant No. 12271104,  51879045).\vspace{2mm}

\noindent{Data Availability Statements}

The data that supports the findings of this study are available within the article.\vspace{2mm}

\noindent{\bf Conflict of Interest}

The authors have no conflicts to disclose.

\end{document}